\documentclass{amsart}
\usepackage{amssymb,amsthm,amsmath}
\newcommand{\bbsm}{\left [\begin{smallmatrix}}      \newcommand{\besm}{\end{smallmatrix}\right ]}
\newcommand{\beq}{\begin{equation}}      \newcommand{\eeq}{\end{equation}}
\newcommand{\beqn}{\begin{eqnarray}}      \newcommand{\eeqn}{\end{eqnarray}}
\newcommand{\beqs}{\begin{eqnarray*}}      \newcommand{\eeqs}{\end{eqnarray*}}
\newcommand\noi{\noindent} 

\newcommand\omg{\Omega}   	\newcommand\bA{\mathbb A} \newcommand\bF{\mathbb F}		 \newcommand\bP{\mathbb P} \newcommand\bZ{\mathbb Z}
\newcommand\bK{\mathbb K}  	\newcommand\mC{\mathcal C}  	 \newcommand\mD{\mathcal D}
\newcommand\mP{\mathcal P}       \newcommand\mH{\mathcal H}	 \newcommand\mG{\mathcal G}
\newcommand\wlk{\wedge^{\ell} \mathbb K^m} 	
\newcommand\wl{\wedge^{\ell}  F^m} 		      \newcommand\wld{\left(\wedge^{\ell} F^m\right)^*}
\newcommand\wlv{\wedge^{\ell} V} 			      
\newcommand\wmlv{\wedge^{m-\ell} V}		      \newcommand\wldv{\wedge^{m-\ell} V^*}
\newcommand\Glm{G(\ell,m)} 		 \newcommand\glm{G(\ell,m)}		       
		 		 \newcommand\gkv{G_{\ell-1}(V)}
\newcommand\glv{G_{\ell}(V)} 		 \newcommand\gmv{G_{\ell+1}(V)}
\newcommand\Fq{{\mathbb F}_q}		 
\newcommand\PAut{\mathop{\rm PAut}}		 \newcommand\MAut{\mathop{\rm MAut}}
\newcommand\GAut{\mathop{\Gamma{\rm Aut}}}		 \newcommand\Aut{\mathop{\rm Aut}}
\newcommand\GL{\mathop{\rm GL}} 	       	\newcommand\gL{\mathop{\Gamma{\rm L}}}
\newcommand\PGL{\mathop{\rm PGL}} 		\newcommand\PgL{{\mathop{{\rm P}\Gamma{\rm L}}}}
 			\newcommand{\C}{C^\mathbb{A}(\ell,m)}
\newcommand\Z{\mathbb Z}  \newcommand\im{\mathop{\rm im}} 
\newtheorem{theorem}{Theorem}[section]
\newtheorem{lemma}[theorem]{Lemma}
\newtheorem{proposition}[theorem]{Proposition}
\newtheorem{corollary}[theorem]{Corollary}
\newtheorem{definition}[theorem]{Definition}

\makeatletter
\def\@tocline#1#2#3#4#5#6#7{\relax
  \ifnum #1>\c@tocdepth 
  \else
    \par \addpenalty\@secpenalty\addvspace{#2}%
    \begingroup \hyphenpenalty\@M
    \@ifempty{#4}{%
      \@tempdima\csname r@tocindent\number#1\endcsname\relax
    }{%
      \@tempdima#4\relax
    }%
    \parindent\z@ \leftskip#3\relax \advance\leftskip\@tempdima\relax
    \rightskip\@pnumwidth plus4em \parfillskip-\@pnumwidth
    #5\leavevmode\hskip-\@tempdima
      \ifcase #1
       \or\or \hskip 1em \or \hskip 2em \else \hskip 3em \fi%
      #6\nobreak\relax
    \dotfill\hbox to\@pnumwidth{\@tocpagenum{#7}}\par
    \nobreak
    \endgroup
  \fi}
\makeatother
\makeatletter
\def\imod#1{\allowbreak\mkern10mu({\operator@font mod}\,\,#1)}
\def\imod#1{\allowbreak\mkern10mu({\operator@font mod}\,\,#1)}
\makeatother

\begin{document}
\title{Automorphism groups of Grassmann codes}
\author{Sudhir R. Ghorpade}    \address{Department of Mathematics,
Indian Institute of Technology Bombay, \newline \indent  Powai, Mumbai 400076, India}
\email{srg@math.iitb.ac.in}
 \author{Krishna V. Kaipa}  \address{Department of Mathematics,
Indian Institute of Science  Education and Research- \newline \indent  Bhopal, Bhauri, Bhopal 462030, India}
\email{kaipa@iiserb.ac.in}

\thanks{Both the authors are partially supported by the Indo-Russian project INT/RFBR/P-114 from the Department of Science \& Technology, Govt. of India. The first named author is also partially supported by an IRCC Award grant 12IRAWD009 from IIT Bombay.} 



\begin{abstract}
We use a theorem of Chow (1949) on line-preserving bijections of Grassmannians to determine the automorphism group of Grassmann codes. Further, we analyze the automorphisms of the big cell of a Grassmannian and then use it to settle an open question of Beelen et al. (2010) concerning the permutation automorphism groups of affine Grassmann codes. Finally, we prove an analogue of Chow's theorem for the case of Schubert divisors in Grassmannians and then use it to determine the automorphism group of linear codes associated to such Schubert divisors. In the course of this work, we also give an alternative short proof of MacWilliams theorem concerning the equivalence of linear codes and a characterization of maximal linear subspaces of Schubert divisors in Grassmannians. 
\end{abstract}

\maketitle

\section{Introduction} \label{sec1}

Let $\mC$ be an $[n,k]_q$-linear code, i.e., $\mC$ be a $k$-dimensional subspace of the $n$-dimensional vector space $\Fq^n$ over the finite field $\Fq$ with $q$ elements, where $q$ is a prime power. Automorphisms of $\mC$ are basically transformations of the ambient space $\Fq^n$ that preserve the code $\mC$ and the coding-theoretic properties of $\mC$. These come in different flavors: permutation automorphisms, monomial automorphisms, and semilinear automorphisms, thus giving rise to groups $\PAut(\mC)$,  $\MAut(\mC)$, and $\GAut(\mC)$ respectively. The last of these, being the most general, will simply be referred to as automorphisms and we will just write $\Aut(\mC)$ instead of $\GAut(\mC)$. Let us recall that $\PAut(\mC)$ consists of permutations $\sigma\in S_n$  such that $(c_{\sigma(1)},\dots,c_{\sigma(n)}) \in C$ for all $c=(c_1,\dots,c_n) \in C$.  Equivalently, $\PAut(\mC)$ consists of permutation matrices $P\in \GL(n, \Fq)$ such that $c P \in \mC$ for all $c\in \mC$. Likewise, 
$\MAut(\mC)$ consists of monomial matrices $M\in \GL(n, \Fq)$ (i.e., matrices of the form $PD$, where $P$ is a permutation matrix and $D$ a diagonal matrix in $\GL(n, \Fq)$) such that $cM\in \mC$ for all $c\in \mC$. Finally, $\Aut(\mC)$ consists of compositions $M\mu$ of (linear maps corresponding to) monomial matrices $M\in \GL(n, \Fq)$ and field automorphisms $\mu$ of $\Fq$ (giving maps of $\Fq^n$ into itself by acting on each of the coordinates) such that $cM\mu\in \mC$ for all $c\in \mC$.  The group $\Aut(\mC)$ of automorphisms of $\mC$ can be viewed as a subgroup of the group $\gL(n,\Fq)$ of semilinear transformations of $\Fq^n$. A classical theorem of MacWilliams \cite{MacWilliams} implies that $\Aut(\mC)$ (resp: $\MAut(\mC)$)  is the same as the group of semilinear (resp: linear) isometries of $\mC$, where by an isometry of $\mC$ we mean a bijection of $\mC$ onto itself that preserve the Hamming metric.

Determination of (either or each of) the automorphism group of a code is of considerable importance. Besides being of  interest in itself, this can be useful in decoding and often finds applications to the theory of finite groups. Complete determination of the automorphism group has been achieved thus far in a few select classes of linear codes, such as Hamming codes, Reed-Muller codes, etc., and we refer to the \emph{Handbook article} of Huffman \cite[Ch. 17]{pless} for more details. See also the more recent work of Berger \cite{Berger} where the case of projective Reed-Muller codes is studied. We consider in this paper the problem of complete determination of the automorphism group(s) of Grassmann codes and also related linear codes such as affine Grassmann codes and Schubert codes. Grassmann codes were introduced by Ryan (1987) in the binary case and by Nogin (1996) in the $q$-ary case, and have been of considerable recent interest (see, e.g., \cite{Nogin, GL, GPP} and the references therein).  However, as far as we know, automorphism groups of Grassmann codes and more generally, of Schubert codes have not been studied so far. Affine Grassmann codes were recently introduced and studied by Beelen, Ghorpade and H{\o}holdt \cite{BGH1, BGH2}. In \cite{BGH1}, it was shown that the permutation automorphism group of the affine Grassmann code $\C$, corresponding to positive integers $\ell, m$ with $\ell < m$, contains a subgroup isomorphic to a semidirect product $M_{\ell\times\ell'}(\Fq) \rtimes \GL_{\ell'}(\Fq)$ of the additive group of $\ell\times\ell'$ matrices over $\Fq$ with the multiplicative group of $\ell'\times\ell'$ nonsingular matrices over $\Fq$, where $\ell'=m-\ell$.
In \cite{BGH2}, this result is extended to show that $\PAut(\C)$ contains a larger group, say ${\mathfrak H}(\ell,m)$, that is essentially obtained by taking the product of the general linear group $\GL_{\ell}(\Fq)$ with the semidirect product $M_{\ell\times\ell'}(\Fq) \rtimes \GL_{\ell'}(\Fq)$. It was also shown in \cite{BGH2} that $\PAut(\C)$ can, in fact, be larger when $\ell'=\ell$, i.e., when $m=2\ell$, and it was remarked that the complete determination of $\PAut(\C)$ is an open question. In this paper, we settle this question and show that  $\PAut(\C)$ coincides with ${\mathfrak H}(\ell,m)$ when $m\ne 2\ell$, whereas if $m=2\ell$, then $\PAut(\C)$ is a certain semidirect product of ${\mathfrak H}(\ell,m)$ with $\Z/2\,\Z$. Moreover, we completely determine $\MAut(\C)$ as well as $\Aut(\C)$. This is, in fact, facilitated by going back to the
origins of affine Grassmann codes, namely, the Grassmann codes, and the corresponding projective systems, viz., (the $\Fq$-rational points of) Grassmann varieties with their canonical Pl\"ucker embedding. Thus, we first tackle the problem of determining  the automorphism group of the Grassmann code $C(\ell,m)$ corresponding to the Grassmannian $G(\ell,m)$ of $\ell$-dimensional subspaces of $\Fq^m$. Here it is actually more natural to deal with the corresponding projective system $\mP = G(\ell,m)$ viewed as a collection of $n = \# G(\ell, m)(\Fq)$ points in the Pl\"ucker projective space $\bP^{k-1} = \bP(\wedge^\ell \Fq^m)$, where $k = \binom{m}{\ell}$. Before proceeding further with the description of our results on automorphisms of Grassmann codes and Schubert codes, it seems worthwhile to digress to discuss the notions of projective systems and their automorphism groups. 

The notion of a projective system was introduced by Tsfasman and Vl\u{a}du\c{t} as an algebraic geometric counterpart of linear codes (see \cite[p. 67]{Tsfasman} for relevant historical and bibliographical information). An $[n,k]_q$-projective system is just a collection, say $\mP$, of $n$ not necessarily distinct points in $\bP^{k-1} = \bP^{k-1}(\Fq)$. The (semilinear) automorphisms of $\mP$ are projective semilinear isomorphisms  (known in classical literature as \emph{collineations}) $g:\bP^{k-1} \to \bP^{k-1}$ that preserve $\mP$ (together with the multiplicities). These form a group that we denote by $\Aut(\mP)$ and call the \emph{automorphism group} of $\mP$; it is a subgroup of the projective semilinear group $\PgL(k, \Fq)$. An $[n,k]_q$-projective system $\mP$ corresponds to an $[n,k]_q$-linear code $\mC$ and $\Aut(\mP)$ corresponds to $\Aut(\mC)$; more precisely, $\Aut(\mC)$ is a central extension of $\Aut(\mP)$ by the subgroup $\Fq^{\times}$ of scalar matrices in $\GL(k,\Fq)$. Likewise, the projective linear  isomorphisms (known in classical literature as \emph{projectivities}) among $\Aut(\mP)$ form a subgroup $\MAut(\mP)$ of $\PGL(k,\Fq)$ that corresponds to $\MAut(\mC)$. Many, but perhaps not all, of the notions and results of classical coding theory can be translated in the language of projective systems and we refer to the book \cite{Tsfasman} for the current state of art. It may be noted, in particular, that a ``Lang-like'' problem \cite[Problem 1.1.9]{Tsfasman} in this book asks to rewrite existing books on coding theory in terms of projective systems, and it is mentioned that the authors consider this to be a rather important and interesting research problem. In the course of our exposition here, we take a small step toward solving this problem by considering automorphisms in the setting of projective systems, and as a dividend, give a short and simple proof of the classical theorem of MacWilliams mentioned earlier. 

Now let us return to a description of our main results. First, we deduce from a theorem of Chow that for the projective system $\mP = G(\ell,m)$ corresponding to Grassmann codes,  $\Aut(\mP)$ is as follows. If $m\ne 2\ell$, then $\Aut(\mP)$ is precisely $\PgL(m,\Fq) = \PgL(\Fq^m)$, viewed as a subgroup of $\PgL(k,\Fq) = \PgL(\wedge^{\ell}\Fq^m)$ by identifying an $m\times m$ nonsingular matrix with its $\ell^{\rm th}$ compound matrix, or in other words, by identifying a linear map $f:\Fq^m\to \Fq^m$ with its $\ell^{\rm th}$ exterior power $\wedge^{\ell} f : \wedge^{\ell}\Fq^m \to \wedge^{\ell}\Fq^m$. On the other hand, if $m=2\ell$, then one has to reckon with a \emph{correlation} and in this case, $\Aut(\mP)$ is a semidirect product of $\PgL(m,\Fq)$ with $\Z/2\Z$. We determine the corresponding group $\Aut(C(\ell,m))$ via central extensions, or more precisely, using group cohomology and certain subgroups of roots of unity in $\Fq^{\times}$; see Theorems \ref{AutP} and \ref{AutC} for a more precise statement. Further, one obtains $\MAut(\mP)$ essentially by replacing $\PgL$ with $\PGL$ in the above description. Next, we take up the case of affine Grassmann codes. To this end, we analyze 
the automorphisms of the so-called big cell of the Grassmannian, and show that they can be lifted to automorphisms of the full Grassmannian and the lifts preserve a certain stratification of the complementary Schubert divisor (Theorem \ref{aut_w0}). Using this and the characterization of the automorphisms of Grassmannians, we can conclude that the semilinear and the monomial automorphism groups of affine Grassmann codes are essentially given by a maximal parabolic subgroup of $\GL(V)$ modulo scalars, although one has to again make a distinction between the cases $m\ne 2\ell$ and $m=2\ell$. (See Theorems \ref{thm:AutAff} and \ref{thm:PAutAff}). Finally, we consider the codes associated to Schubert divisors in Grassmannians; these are a special case of Schubert codes introduced in \cite{GL} and are in a sense complementary to affine Grassmann codes. We show, in fact, that the automorphism groups of these codes are isomorphic to that of affine Grassmann codes, and that the automorphisms of Schubert divisors in a Grassmannian can be identified with the automorphisms of the big cell of the Grassmannian.  However, proving this in the setting of an arbitrary ground field (and in particular, a finite field) seems quite nontrivial and a number of auxiliary results are needed. In particular, we require a classification of the maximal linear subspaces of Schubert divisors. In fact, we determine the maximal linear subspaces of the Schubert divisor as well as those of its big cell in Lemmas ~\ref{max_lin_schubert}, and ~\ref{lin_W_1}, and these may be of independent interest.

A key ingredient in our determination of the automorphism group of Grassmann codes is a classical result of Chow \cite{Chow} (Theorem~\ref{chow_thm}) that characterizes line-preserving bijections of Grassmannians and can be viewed as a remarkable generalization of the fundamental theorem of projective geometry. This result has also been rediscovered by others (e.g., Nemitz \cite{Nemitz}) and it has been extended to the case when the base field is replaced by a division ring (see, e.g., \cite{Wan}). For a proof of Chow's theorem, one may refer to the books of Pankov \cite{Pankov} and Wan \cite{Wan}. The result mentioned earlier about  
the automorphisms of the projective system $\mP = G(\ell,m)$ corresponding to Grassmann codes is, in fact, just a paraphrasing of Chow's theorem. However, the corresponding result about the automorphisms of Schubert divisors in Grassmannians appears to be new.

\section{Projective Systems and Automorphisms of Codes}     \label{sec2}

In this section we consider the notions of equivalence and automorphisms  of linear codes from the point of view of  projective systems.  The purpose is to settle notations and terminology used throughout this paper. For a more leisurely treatment one may refer to \cite{Tsfasman}. Fix a prime power $q$ and let $F$ denote the finite field $\bF_q$. Let $\mC \subset F^n$  denote a nondegenerate $[n,k]_q$-linear code. The standard basis of $F^n$ will be denoted $\{e_1, \dots, e_n\}$ and the associated dual basis will be denoted $\{e^1, \dots, e^n\}$. The nondegeneracy of $\mC$ implies that the restrictions of the functionals $e^1, \dots, e^n$ to $\mC$ are nonzero and they span $\mC^*$. Thus if we let $\mP$ denote the collection of $n$ points of $\bP(\mC^*)$  corresponding to the restrictions 
$e^1_{|\mC}, \cdots, e^n_{|\mC}$, then $\mP$ consists of $n$ (not necessarily distinct) points in $\bP(\mC^*)$ not lying in a hyperplane of the projective space $\bP(\mC^*)$. We say that $\mP \subset \bP(\mC^*)$ is the \emph{projective system} associated with $\mC$. Note that $\dim(C^*) =k$ and thus $\bP(\mC^*)$ may be identified with $\bP^{k-1}_{\Fq}$. 

If $\mD \subset \mC$ is an $r$-dimensional subcode of $\mC$, then the subspace  of $\mC^*$ consisting of functionals vanishing on $\mD$ has codimension $r$. Let $L_{\mD} \subset \bP(\mC^*)$ be its projectivization, where, by convention, $L_{\mC} = \varnothing$. Using this correspondence $\mD \mapsto L_{\mD}$ we identify $r$-dimensional subcodes of $\mC$ with codimension $r$ subspaces of $\bP(\mC^*)$. The weight of an $r$-dimensional subcode $\mD$, denoted $\Vert \mD \Vert $, is the cardinality $|\mP \!\setminus\! L_{\mD}|$.  In other words, $\Vert \mD \Vert $ is the number of elements of $\mP$ (counting multiplicities) that are not in $L_{\mD}$. The weight of a nonzero codeword $v \in \mC$, denoted $\Vert v \Vert$,  is the weight of the one-dimensional subcode generated by it.  
We recall the following elementary formula \cite[p. 7]{Tsfasman} relating the weight of a subcode to that of its one-dimensional subcodes:  
\begin{equation} \label{eq:||D||}
\Vert \mD \Vert  = \frac{1}{q^{r-1}} \sum_{[v] \in \bP(\mD)} \Vert v \Vert . 
\end{equation}
Geometrically,  $\Vert \mD \Vert = |\mP \!\setminus\! L_{\mD}|$, and   $\mP \!\setminus\! L_{\mD} = \cup_{\mathcal H \supset L_{\mD}} \mP \!\setminus\! \mathcal H $, where the union is over hyperplanes $\mathcal H$ in $\bP(\mC^*)$ containing $L_{\mD}$. The  formula \eqref{eq:||D||} now follows from the fact that  for each 
 $x \in \mP \! \setminus \! L_{\mD}$, there are $q^{r-1}$ hyperplanes containing $L_{\mD}$ but not $x$.

We recall the notion of  a semilinear map $f:V \to V'$ between vector spaces over any field $\bK$. We say $f$  is   $\mu$-semilinear  if there is a field automorphism $\mu$ of $\bK$ with $f(u+c v) = f(u)+\mu(c) f(v)$ for all $u,v \in V$ and $c \in \bK$. In particular an invertible $\mu$-semilinear transformation  of $\bK^m$ has  the form $x \mapsto A \, \mu(x)$ where $A \in \GL(m,\bK)$ and $\mu$ acts on each entry of $x$.  The group of  semilinear automorphisms of a vector space $V$ will be denoted  $\gL  (V)$, and also by  $\gL (m,\bK)$ when $V = \bK^m$. For a semilinear map $f:V \to V'$, we denote the corresponding projective semilinear map from $\bP(V)$ to $\bP(V')$ by $[f]$.  By a projective semilinear isomorphism from $\bP(V)$ to $\bP(V')$ we mean the map induced by a semilinear isomorphism from $V$ to $V'$.  We will denote the group of projective semilinear isomorphisms  from $\bP(V)$ onto itself by  $\PgL(V)$, and we note that $\PgL(V)$ is the factor group of $\gL (V)$ by the subgroup $\bK^{\times}$ of transformations of $V\to V$ of the form $ v \mapsto c v$  as $c$ varies over the nonzero elements of $\bK$. 
The semilinear isomorphisms of finite dimensional projective spaces are characterized by the following  (see, e.g., \cite[Theorem 2.26]{Artin}):

\begin{theorem}[Fundamental Theorem of Projective Geometry]  \label{fund_thm}
Let $n \geq 3$. The group of  bijective self maps of the projective space $\bP(\bK^n)$ which take lines to lines, is  the group $\PgL (n, \bK)$.
\end{theorem}

We will also need the notion of the transpose of a $\mu$-semilinear map $f:V \to V'$ between vector spaces over any field $\bK$. The \emph{transpose}  $f^*:V'^* \to V^*$  is the unique semilinear map satisfying the property: 
\[
\langle \alpha, f(v) \rangle = \mu (\langle f^* \alpha, v \rangle) \quad \text{ for all } \alpha \in V'^* \text{ and } v \in V,
\]  
where $\langle ~,~ \rangle$ denotes the natural pairing $V^* \times V \to \bK$. If $f$ is represented as $x \mapsto A \mu(x)$ in coordinates with respect to some choice of bases of $V$ and $V'$, then the transpose is  represented in coordinates with respect to the dual bases by $\xi \mapsto \mu^{-1} (A^t \xi)$ where $A^t$ is the usual transpose and $\mu^{-1}$ acts entry-wise on matrices. For a projective semilinear isomorphism $g:\bP(V) \to \bP(V')$, we define a transpose as follows. Let $f: V \to V'$ be any semilinear map such that $[f]=g$, and let $f^*:V'^* \to V^*$ be its transpose. Then the projective map $[ f^*]:\bP(V'^*) \to \bP(V^*)$ induced by $f^*$ is independent of the choice of $f$ and it will be denoted by $g*$ and called the \emph{transpose} of $g$. 

We recall (from \cite{Tsfasman}) the notions of equivalence and isomorphisms for  codes and for projective systems. We allow  all maps to be semilinear.

\begin{definition}  \label{equivalence}
{\rm
Let $\mC, \mC'$ be nondegenerate $q$-ary linear codes and let $\mP \subset \bP (\mC^*)$ and $\mP' \subset \bP (\mC'^*)$ be the corresponding projective systems. We say that $\mP$ and $\mP'$ are \emph{equivalent} if there is a projective semilinear isomorphism $g: \bP(\mC^*) \to \bP(\mC'^*)$ which carries $\mP $ to $\mP'$.
We say that the codes $\mC$ and $\mC'$ are  \emph{equivalent}, and write  $\mC \sim \mC'$, if the corresponding projective systems $\mP$ and $\mP'$ are equivalent. 
}
\end{definition}

Note that since the points of $\mP, \mP'$ need not be distinct, when we say $g$ carries $\mP$ to $\mP'$, we mean that for each $p\in \mP$, the multiplicity of $g(p)$ in $\mP'$ equals the multiplicity of $p$ in $\mP$. In particular, if two codes are equivalent, then their lengths and dimensions coincide. Suppose both $\mP$ and $\mP'$ consist of $n$ elements, i.e., the length of $\mC$ as well as $\mC'$ is $n$. Let Isom$(F^n)$ be the subgroup of $\gL (n,F)$ consisting of transformations preserving the Hamming metric on $F^n$. They are of the form $x \mapsto A \mu(x)$   where $A$ is a monomial matrix and $\mu$ a power of the Frobenius automorphism. If $\mC \sim \mC'$ with $g: \bP(\mC^*) \to \bP(\mC'^*)$ as above, then there is a unique $g^{\flat} \in $ Isom$(F^n)/F^{\times}$ that takes $\mC$ to $\mC'$. 
To see this, let $\tilde g:\mC^* \to \mC'^*$ be a semilinear isomorphism inducing $g$, then $g(\mP) = \mP'$ implies that $\tilde g$ takes  
$e^j_{|\mC}$ to $a_j \, e^{\sigma(j)}_{|\mC'}$ for some permutation  $\sigma \in S_n$ and some $a_j \in F^{\times}$.  If $\tilde g$ is $\mu$-semilinear, then there is a unique $\mu$-semilinear automorphism $\phi^*$ of $(F^n)^*$ whose restriction to $C^*$ is $\tilde g$. It is given by $\phi^*(\xi) = A \mu(\xi)$ where  $A$ is the monomial matrix in $\GL(n, F)$ whose $(i,j)$th entry is $a_j \delta_{i,\sigma(j)}$. Thus $ (\tilde g^{-1})^*: \mC \to \mC'$ is the restriction to $\mC$ of the isometry $\phi^{-1}$ of $F^n$ that takes $x \mapsto  B \mu(x)$ where  $B=(B_{ij}) \in \GL(n, F)$ is given by $B_{ij} = (1/a_j)  \, \delta_{i,\sigma(j)}$. Since $\tilde g$ was ambiguous up to a scalar multiple, so is the isometry $\phi^{-1}$. The class of $\phi^{-1}$ in Isom$(F^n)/F^{\times}$ is the desired map $g^{\flat}$. 

\begin{definition}  \label{aut_def}
{\rm 
Let $\mC \subset F^n$  be a nondegenerate $[n,k]_q$-linear code, and $\mP$  the corresponding projective system.   
The \emph{automorphism group}  $\Aut(\mC)$ of the code $\mC$ is the group of transformations of $F^n$ onto itself that take $\mC$ to $\mC$   
and have the form $x \mapsto A \mu(x)$,  where $A$ is a monomial matrix in $\GL(n,F)$ and $\mu$ a field automorphism of $F$.
The \emph{automorphism group} $\Aut (\mP)$ of the projective system $\mP$ is  the subgroup of $\PgL(\mC^*)$ of transformations taking $\mP$ to itself. 
}
\end{definition}

We remark that for $\mC$ and $\mP$ as above,  $\Aut (\mC)$ is a central extension of $ \Aut (\mP)$ by $F^{\times}$, i.e., $\Aut (\mP) \simeq \Aut (\mC)/F^{\times}$, where $F^{\times}$ is the subgroup of scalar matrices in $\GL(n,F)$. The isomorphism between $\Aut(\mP)$ and $\Aut(\mC)/F^{\times}$ is given by the correspondence $g \mapsto g^{\flat}$ described above. 

In addition to the equivalence relation $\sim$ of codes (Definition \ref{equivalence}), there is another  natural notion of equivalence of nondegenerate $[n,k]_q$ codes:  

\begin{definition}  \label{approx_def} 
{\rm
Let $\mC, \mC'$ be nondegenerate $q$-ary linear codes. Define $\mC \approx \mC'$ if there is a weight preserving bijection $f:\mC \to \mC'$ carrying  $r$-dimensional subcodes of $\mC$ to $r$-dimensional subcodes of $\mC'$ for all $r\ge 0$.
}
\end{definition}

We now give a geometric proof of the following well-known result of MacWilliams. The proof below seems much shorter than the standard proofs (cf. \cite{MacWilliams,Bogart,Ward}) and illustrates the advantage of the geometric framework of projective systems.

\begin{theorem}[MacWilliams]  \label{MacWilliams}
Let $\mC, \mC' $ be nondegenerate $q$-ary linear codes of dimension $\geq 3$. Then:  $\;\mC \sim \mC' $   if and only if $\mC \approx \mC'$.
\end{theorem}

\begin{proof} 
Suppose $\mC \approx \mC'$ and let $f:\mC \to \mC'$  be as in the Definition~\ref{approx_def}.  Clearly, $\dim \mC = \dim \mC' = k $ (say). 
Since the correspondence $\mD \to L_{\mD}$ maps $(k-1)$-dimensional subcodes of $\mC$ to codimension $k-1$  subspaces of $\bP(\mC^*)$, viz., points of $\bP(\mC^*)$, we obtain from $f$  an induced bijective map $f_{\sharp}:\bP(\mC^*) \to \bP(\mC'^*)$. Given an $r$-dimensional subcode $\mD$ of $\mC$, the map $f$ carries the set of all $(k-1)$-dimensional subcodes of $\mC$ that contain $\mD$ into the set of $(k-1)$-dimensional subcodes of $\mC'$ containing $f(\mD)$. This implies that $f_{\sharp}$ carries codimension $r$ subspaces to codimension $r$ subspaces. Since $k \geq 3$, Theorem~\ref{fund_thm} implies that $f_{\sharp}$ is a projective semilinear isomorphism. Moreover,  from formula \eqref{eq:||D||}, it follows that $f$ preserves weights of all subcodes. In particular, since points of $\bP (\mC^*)$ correspond to $(k-1)$-dimensional subcodes, we see that  $|\{v\} \cap \mP| = |\{f_{\sharp}(v)\} \cap \mP'|$, for all $v \in \bP(\mC^*)$. Therefore $P \in \mP$ with multiplicity $\nu$ if and only if $f_{\sharp}(P) \in \mP'$ with multiplicity $\nu$. Thus the semilinear isomorphism $f_{\sharp}$ carries $\mP$ to $\mP'$, and so $\mC \sim \mC'$. For the converse, the discussion preceding Definition~\ref{aut_def} shows that $\mC \sim \mC'$ implies the existence of a transformation $\phi^{-1} \in$ Isom$(F^n)$ of $F^n$ taking $\mC$ to $ \mC'$.  Since Hamming isometries preserve weights of codewords,  $\phi_{|\mC}$ gives the equivalence $\mC \approx \mC'$.
\end{proof}

We observe that if $\mC \approx \mC'$ with $f:\mC \to \mC'$ a linear isometry, then the map $f_{\sharp}$ is automatically a projective linear isomorphism, and hence we do not require the $k \geq 3$ assumption necessary for Theorem \ref{fund_thm}. We state this as a corollary:

\begin{corollary} 
If $f:\mC \to \mC'$ is a linear isometric isomorphism between two $q$-ary linear codes, then $\mC$ and $\mC'$ have the same length $n$ and $f$ is the restriction to $\mC$ of a linear transformation of $F^n$ defined by a monomial matrix in $\GL(n,F)$.
\end{corollary}

\section{Grassmann Codes}  \label{sec3}

In the first subsection below, we work over an arbitrary field $\bK$. Subsequently, we will let $\bK$ be the finite field $F=\Fq$.   Fix positive integers $\ell, m$.  
Let $\Glm$ denote the set of $\ell$-dimensional subspaces of $\bK^m$. We will also use the symbol $V$ for $\bK^m$, and denote $\Glm$ by $G_{\ell}(V)$.
Let $\{e_1, \dotsc, e_m\}$ be a basis of $\bK^m$ and let $\mathcal I(\ell,m)$ denote the set of multi-indices:
\[ 
\mathcal I(\ell,m) = \{ (i_1,\cdots,i_{\ell}) : 1 \leq i_1 < i_2 < \dotsb < i_{\ell} \leq m\}. 
\]
For $I=(i_1,\cdots,i_{\ell})\in \mathcal I(\ell,m)$, let $e_I$ denote the decomposable element $e_{i_1}\wedge \dots \wedge e_{i_{\ell}}$ of $\wlv$. 
Then $\{ e_I  :  I \in  \mathcal I(\ell,m) \}$ is a basis of $\wlv$.  To any  $\gamma \in \Glm$, we assign the wedge product  $v_1 \wedge v_2 \wedge \dotsb \wedge v_{\ell} \in \bP(\wlv)$ where $\{ v_1, \dotsc, v_{\ell}\}$ is an arbitrary basis of $\gamma$. Writing  $v_1 \wedge v_2 \wedge \dotsb \wedge v_{\ell} = \sum_I p_I e_I$, with $p_I\in \bK$ for $I\in \mathcal I(\ell,m)$, we obtain the Pl\"{u}cker coordinates  $\left(p_I\right)_{I\in \mathcal I(\ell,m)}$ of $\gamma$. It is well-known that the 
resulting map $\Glm\to \bP(\wlv)$ is a nondegenerate embedding of $\Glm$ as a projective subvariety of $\bP(\wlv)$ defined by certain quadratic polynomials  (see, e.g., \cite[\S VII.2]{HP1}). We will refer to $\Glm$ as the \emph{Grassmannian} or the \emph{Grassmann variety} (of $\ell$-dimensional subspaces of $V=\bK^m$). Note that if $\ell =m$, then  $\Glm$ reduces to a point, whereas if $\ell=1$ (or $\ell=m-1$), then $\Glm$ is just the projective space $\bP(V)$. With this in view, to avoid trivialities, we shall henceforth assume that $1<\ell<m$. This means, in particular, that $m\ge 3$. 

\subsection{Line-preserving Bijections of Grassmannians}
By a \emph{line} in $\glv$  we mean a set of $\ell$-dimensional spaces of $V$ containing a given $(\ell-1)$-dimensional subspace and contained in a given $(\ell+1)$-dimensional subspace (see the discussion following Lemma \ref{max_lin}). We define: $\Aut(\glv)$ to be the set of all bijections $f:G_{\ell}(V)  \to G_{\ell}(V)$  such that both $f$ and $f^{-1}$ take lines to lines. Evidently $\Aut(\glv)$ is a group with respect to composition. The group $\Aut(\glv)$ was explicitly determined by Wei-Liang Chow (1949); see Theorem~\ref{chow_thm} below. We will need a few definitions and lemmas (which are also needed in later sections) before we can state this theorem. Note that the assumption $\ell < m$ is crucial in the lemma below. 

\begin{lemma} \label{rho_def}
Let $\hat \rho: \gL (V) \to  \gL (\wlv)$ be the $\ell$-th  exterior power representation that maps $f \in  \gL (V)$ to $\wedge^{\ell} f$, where $\wedge^{\ell} f \in \gL (\wlv)$ is defined by 
\[
(\wedge^{\ell} f)(v_1 \wedge \cdots \wedge v_{\ell})=f(v_1)   \wedge \cdots \wedge f(v_{\ell}) \quad \text{ for } v_1, \dots , v_{\ell}\in V. 
\]
Then the induced homomorphism $\rho: \PgL (V) \to  \PgL (\wlv)$  is injective. Moreover, the kernel of $\hat \rho$ is 
\begin{equation*} 
\label{eq:kerrho}
\ker (\hat \rho) = \{ c \, I_m \,:\, c\in \bK \text{ with } c^{\ell} = 1\}. 
\end{equation*}
\end{lemma}

\begin{proof}
For  $\gamma \in G_{\ell}(V)$ and $g \in \PgL (V)$, denote by $g(\gamma) \in G_{\ell}(V)$ the image  under $g$ of the subspace $\gamma$ of $V$. 
Let $g \in \ker(\rho)$. Then for any  $\gamma \in G_{\ell}(V)$, the  Pl\"{u}cker coordinates of  $g(\gamma)$ and $\gamma$ are equal, and so $g(\gamma) =\gamma$.  
Suppose there exists $v \in \bP(V)$ with $g(v) \neq v$. 
Since $\ell < m$, we can find $\gamma \in G_{\ell}(V)$ 
containing $v$ such that $g(v) \notin \gamma$. 
But $v \in \gamma$ implies $g(v) \in g(\gamma) = \gamma$, which is a contradiction. 
This shows that 
$g$ is the identity element of $\PgL (V)$. Thus $\rho$ is injective.  
It follows that the elements of ker$(\hat \rho)$  are scalar matrices, and clearly $c \, I_m \in$ ker$(\hat \rho)$  if and only if $c^{\ell}=1$. 
\end{proof}


We recall that the group $\PGL(m,\bK)$ has an inverse transpose automorphism taking $[A] \mapsto [A^{-t}]$, where $A \in \GL(m,\bK)$ is a representative of $[A]$. For $m \geq 3$, this is an outer automorphism corresponding to a diagram   automorphism of the root system. (For $m=2$, the automorphism is inner because $[A^{-t}] = \bbsm 0 & -1\\ 1 &0  \besm [A] \bbsm 0 & -1\\ 1 &0  \besm^{-1}$). There is an induced automorphism on $\gL (V)$ taking a $\mu$-semilinear  map $x \mapsto A \mu(x)$ to the $\mu$-semilinear map $x \mapsto A^{-t} \mu(x)$, which in turn descends to an automorphism of $\PgL (V)$. We denote these automorphisms of 
$\gL(V)$ and $\PgL(V)$ by a common symbol $-t$, and refer to it as the \emph{inverse transpose outer automorphism}.

\begin{definition} \label{Hodge_def}
{\rm
The \emph{Hodge star isomorphism}   $\ast_{\ell}: \wlv \to \wmlv$ is the linear isomorphism of $\wlv$ defined by $e_I \to $ \emph{sgn}$(I I^{\circ})\,  e_{I^{\circ}}$ where $I \in \mathcal I(\ell,m)$ and  $I^{\circ} \in \mathcal I(m-\ell,m)$ is the complement of $I$. The scalar \emph{sgn}$(I I^{\circ}) \in \{ \pm 1\}$ denotes the sign of the permutation $(1,\cdots,m)  \mapsto (I I^{\circ})$. We use the same symbol $\ast_{\ell}:\bP(\wlv) \to \bP(\wmlv)$ for the induced projective linear isomorphism.
}
\end{definition}

\begin{proposition}  \label{Hodge_prop}
Let $V = \bK^m$.  The isomorphism $\ast_{\ell}$ takes $G_{\ell}(V)$ to $G_{m-\ell}(V)$. Moreover, $\ast_{\ell} \circ \wedge^{\ell} [f] \circ  \ast_{\ell}^{-1}= \wedge^{m-\ell} [f]^{-t}$ for all $[f] \in \PgL (V)$. In the case  $m=2 \ell$, $ \wedge^{\ell} [f] \mapsto \ast_{\ell} \circ \wedge^{\ell} [f] \circ  \ast_{\ell}^{-1}$  is  a nontrivial  outer automorphism of \emph{im}$(\rho)$.
\end{proposition}

\begin{proof} Let $f \in \gL (V)$ be a representative of $[f]$, and let $x \mapsto A \mu(x)$ denote the coordinate expression of $f$ with respect to the basis $\{e_1, \dots, e_m\}$. We will show $\ast_{\ell}\circ \wedge^{\ell} f \circ  \ast_{\ell}^{-1} = {\rm det}(A) \wedge^{m-\ell} f^{-t}$ holds $\gL (V)$, which in turn implies that 
$\ast_{\ell} \circ \wedge^{\ell} [f] \circ  \ast_{\ell}^{-1}= \wedge^{m-\ell} [f]^{-t}$. We construct linear isomorphisms $H_{\ell}: \wlv \to \wldv$ and $\theta:V^*  \to V$ such that $\ast_{\ell} =  (\wedge^{m-\ell} \theta) \circ H_{\ell} :\wlv \to \wmlv$.  Let $\eta = e_1 \wedge \dots \wedge e_m \in \wedge^m V$. For any $\xi \in \wlv$, let $H_{\ell}(\xi) \in \wldv$ be defined by $\langle H_{\ell}(\xi), \zeta \rangle \,   \eta  =  \xi \wedge \zeta $  for all $\zeta \in \wmlv$, where $\langle ,\rangle$ is the natural pairing between  $\wedge^i V^*$ and $\wedge^i V$ for each $0\leq i \leq m$. We claim that $H_{\ell}$ carries a decomposable element $v_1 \wedge \dots \wedge v_{\ell}$ of $\wlv$ to a decomposable element of $\wldv$. To see this, we extend $\{v_1, \dots, v_{\ell}\}$ to a basis $\{v_1, \dots, v_m\}$ of $V$ such that $v_1 \wedge \dots \wedge v_m = \eta$. It follows from the definition of $H_{\ell}$ that $H_{\ell}(v_1 \wedge \dots \wedge v_{\ell})$ is the decomposable element  $v^{\ell+1} \wedge \dots \wedge v^m$, where $\{v^1, \dots,v^m\}$ is the dual basis. In particular $H_{\ell}(e_I) = {\rm sgn}(II^{\circ}) e^{I^{\circ}}$.
Now let $\theta:V^*  \to V$ be the linear isomorphism defined by $e^i \mapsto e_i$ for all $i$. It now follows that $(\wedge^{m-\ell} \theta) \circ H_{\ell}$ carries $e_I$ to sgn$(I I^{\circ})\,  e_{I^{\circ}}$, and hence  $(\wedge^{\ell} \theta) \circ H_{\ell}$ is $\ast_{\ell}$. The map $\wedge^{m-\ell} \theta$ takes a decomposable element  $v^1 \wedge \dots \wedge v^{m-\ell} \in \wldv$ to the decomposable element $\theta(v^1) \wedge \dots \wedge \theta(v^{m-\ell})$. Thus the composition $(\wedge^{m-\ell} \theta) \circ H_{\ell}$ takes  decomposable elements to decomposable elements . Therefore the projective map $\ast_{\ell}$ takes $G_{\ell}(V)$ to $G_{m-\ell}(V)$.

Next, we observe that the maps $f^{-t}$ and $\theta \circ (f^{-1})^* \circ \theta^{-1}$ are equal, as both have the coordinate expression $x \mapsto A^{-t} \mu(x)$ 
(the transpose $f^*$ of a semilinear map $f$ was defined immediately after Theorem \ref{fund_thm}). Using the definition $\ast_{\ell} = (\wedge^{m-\ell} \theta) \circ H_{\ell}$ we get,  $\ast_{\ell} \circ \wedge^{\ell} f \circ  \ast_{\ell}^{-1} = \wedge^{m-\ell} \theta \circ H_{\ell} \circ \wedge^{\ell} f \circ {H_{\ell}}^{-1} \circ \wedge^{m-\ell} \theta^{-1}$. Therefore, it suffices to show $ H_{\ell} \circ \wedge^{\ell} f \circ {H_{\ell}}^{-1} = {\rm det}(A) \wedge^{m-\ell} (f^{-1})^*$. 
To prove this we need to show: 
\beq  \label{eq:Hodge_step1}
\langle (H_{\ell} \circ \wedge^{\ell} f) (\xi), \zeta \rangle =   {\rm det}(A) \, \langle (\wedge^{m-\ell} (f^{-1})^* \circ H_{\ell}) (\xi), \zeta \rangle
\eeq
holds for all $\xi \in \wlv, \,\zeta \in \wmlv$. Multiplying both sides of \eqref{eq:Hodge_step1} by $\eta$, the left side is $(\wedge^{\ell} f (\xi)) \wedge \zeta$  which we rewrite as:
\beq  \label{eq:Hodge_step2}
 \wedge^m\!f \,\left( \xi \wedge (\wedge^{m-\ell}f^{-1})(\zeta) \right) = 
\mu\left( \langle  H_{\ell}(\xi), \wedge^{m-\ell} f^{-1} (\zeta) \rangle \right) \, {\rm det}(A) \, \eta  \eeq
which is $\eta$ multiplied by the right side of \eqref{eq:Hodge_step1}, as desired. 

Finally, suppose $m = 2 \ell$. Note that $m\ge 3$ since $\ell >1$. Assume, on the contrary, that  $\wedge^{\ell} [f] \mapsto \ast_{\ell} \circ \wedge^{\ell} [f] \circ  \ast_{\ell}^{-1}$ is an inner automorphism of $\im (\rho)$, obtained by conjugating by $\wedge^{\ell} [f_0]$. Then, using  $\ast_{\ell}\circ \wedge^{\ell} [f] \circ  \ast_{\ell}^{-1} = \wedge^{\ell} [f]^{-t}$, we get  $\rho([f_0 \circ f \circ f_0^{-1}]) = \rho([f]^{-t})$. Since $\rho$ is injective, we get a contradiction to the fact that $-t$ is not an inner automorphism of $ \PgL (V)$.
\end{proof}

\begin{corollary}   \label{Hodge_crlry}
Assume that $m=2\ell$. Consider the semidirect product
\[
\PgL (V) \rtimes_{-t} \mathbb Z / 2 \mathbb Z:=\langle \PgL (V), \epsilon \, |\, \epsilon^2=I_{\PgL(V)}, \epsilon g \epsilon^{-1} = g^{-t} \rangle. 
\]
Let  $\emph{im} (\rho) \rtimes_{\ast_{\ell}} \mathbb Z / 2 \mathbb Z$ be the subgroup of $\PgL(\wlv)$ generated by im$(\rho)$ and $\ast_{\ell}$:
\[
\emph{im} (\rho) \rtimes_{\ast_{\ell}} \mathbb Z / 2 \mathbb Z:= \langle \emph{im} (\rho) , \ast_{\ell} \, |\, \ast_{\ell}^2=I_{\PgL(\wlv)}, \ast_{\ell} (\wedge^{\ell}g) \ast_{\ell}^{-1} = \wedge^{\ell} g^{-t}, \rangle. 
\]
where $I_G$ denotes the identity element of a group $G$. The homomorphism  from $\PgL (V) \rtimes_{-t} \mathbb Z / 2 \mathbb Z$  to $\emph{im} (\rho) \rtimes_{\ast_{\ell}} \mathbb Z / 2 \mathbb Z$ that restricts to $\rho$ on $\PgL(V)$ and sends the generator $\epsilon$ to $\ast_{\ell}$ is an isomorphism.
\end{corollary} 

\begin{proof} 
We note that  $\ast_{\ell}^2 = (-1)^{\ell(m-\ell)} I_{\gL(\wl)}$, and hence when $m = 2 \ell$ we get  $\ast_{\ell}^2 = (-1)^{\ell} I_{\gL(\wlv)}$ in $\gL(\wlv)$ and $\ast_{\ell}^2 = I_{\PgL(\wlv)}$ in $\PgL(\wlv)$. Therefore the homomorphism described in the statement is well defined. It is clearly surjective, and it is injective because  $\rho$ is injective and  $\ast_{\ell} \notin {\rm im}(\rho)$ (because $\wedge^{\ell} g \mapsto \ast_{\ell} \circ \wedge^{\ell} g \circ  \ast_{\ell}^{-1}$ is not an inner automorphism of $\im (\rho)$).
\end{proof}

\noi {\bf Remark:} We will also need the following relation between $\ast_{\ell-1}$ and $\ast_{\ell}$, which follows from their definitions. 
\beq \label{eq:star_prop} 
\ast_{\ell}( \beta \wedge v) = \iota_{\theta^{-1}(v)} (\ast_{\ell-1} \beta)   \quad \mbox{for all} \quad \beta \in \wedge^{\ell-1} V, \, v \in V,
\eeq 
where $\iota_{\omega}:\wedge^{r} V \to \wedge^{r-1} V$ is the interior multiplication operation defined by  $\langle \nu, \iota_{\omega} \xi\rangle = \langle \omega \wedge \nu, \xi \rangle$ for all  $\nu \in \wedge^{r-1} V^*$, $\omega \in V^*$ and $\xi \in \wedge^r V$. 
\smallskip

We use the letters $\alpha, \beta, \gamma$ and $\delta$ for points of $G_{\ell-2}(V), G_{\ell-1}(V),G_{\ell}(V)$, and $G_{\ell+1}(V)$ respectively.  By a linear subspace of $G_{\ell}(V)$, we mean a linear subspace of $\bP(\wlv)$ contained in $G_{\ell}(V)$. The following basic fact characterizes the linear  subspaces of $G_{\ell}(V)$ (cf.  \cite[\S II.1]{Chow}, \cite[\S 24.2]{HT}, \cite[p. 68]{Harris}, \cite[\S 2]{GPP}). 

\begin{lemma}[Maximal linear subspaces of $G_{\ell}(V)$] \label{max_lin} 
For any $\beta \in G_{\ell-1}(V)$ and any $\delta \in G_{\ell+1}(V)$, let 
\[
 \pi_{\beta} = \{ \gamma \in G_{\ell}(V) : \gamma \supset \beta\} \quad \text{and} \quad 
 \pi^{\delta} = \{ \gamma \in G_{\ell}(V)  : \gamma \subset \delta\}.
\] 
Then $\pi_{\beta}$ and $\pi^{\delta}$ are linear subspaces of $G_{\ell}(V),$ of dimensions $m-\ell$ and $\ell$ respectively. Moreover, any linear subspace of $G_{\ell}(V)$ is contained in a $\pi_{\beta}$ or a $\pi^{\delta}$. 
\end{lemma}

We note that the intersection of all $\gamma$  contained in a $\pi_{\beta}$  is $\beta$. Therefore, $\beta \to  \pi_{\beta}$ is one-one. Similarly the vector space sum of all $\gamma$ in a $\pi^{\delta}$ is $\delta$, therefore $\delta \to  \pi^{\delta}$ is one-one. In case  $m=2 \ell$, both $\pi_{\beta}$ and $\pi^{\delta}$ are $\ell$-dimensional, but $\pi_{\beta} \neq \pi^{\delta}$ because the intersection  of all $\gamma$ in  a $\pi_{\beta}$ is $\beta$, where as it is just $\{0\}$ for a $\pi^{\delta}$.  We note that  $\pi^{\delta}$ and $\pi_{\beta}$  are projectively isomorphic to  $\bP(\delta^*)$ and  $\bP(V/\beta)$ respectively. Therefore the linear subspaces of $\pi^{\delta}$ and $\pi_{\beta}$ (of codimension $r$) are:
\beqn \label{eq:lin_glv}
\pi_W^{\delta} := \{ \gamma \in G_{\ell}(\delta)  :  \gamma \supset W\} \quad \text{for  } W \in G_{r}(\delta), \text{ and} \\
\nonumber  \pi_{\beta}^U := \{ \gamma \in \glv  : \beta \subset \gamma \subset U\} \quad  \text{for  } U\in G_{m-r}(V) \text{ with }
\beta \subset U, 
\eeqn
respectively. 
In particular, the lines contained in a $\pi^{\delta}$ or in a $\pi_{\beta}$ are of the form $\pi_\beta^\delta$ for $\beta \in G_{\ell-1}(\delta)$. 
It follows that every line in $\glv$ is of the form $\pi_{\beta}^{\delta}$. We also note that $\pi_{\beta} \cap \pi^{\delta}$ is the line $\pi_{\beta}^{\delta}$ if $\beta \subset \delta$, and is empty otherwise. Similarly $\pi_{\beta} \cap \pi_{\beta'}$ is the point $\beta+\beta'$ or  empty depending on whether $\beta, \beta'$ are collinear in $\gkv$ or not, and $\pi_{\delta} \cap \pi_{\delta'}$ is the point $\delta \cap \delta'$  or empty depending on whether $\delta, \delta'$ are collinear in $\gmv$ or not.
We now state an important theorem of Chow which will be used in the sequel.

\begin{theorem}[Chow]
\label{chow_thm} Let $V = \bK^m$ with $m \geq 3$. Then 
\[ \emph{Aut}(\glv) =   \begin{cases}      {\rm im}(\rho) \simeq \PgL (V)     & \mbox{ if }   m \neq 2 \ell, \\
{\rm im}(\rho) \rtimes_{\ast_{\ell}} \mathbb Z / 2 \mathbb Z \simeq \PgL (V) \rtimes_{-t} \mathbb Z / 2 \mathbb Z  & \mbox{ if }  m=2 \ell.  \end{cases} \]
\end{theorem}

\noi Proofs of this theorem are given in Chow \cite[\S II.1]{Chow}, Pankov  \cite[Theorem 3.2]{Pankov},  Wan \cite[Theorem 3.45]{Wan} and Nemitz \cite{Nemitz}, with two differences.  The first is that Aut$(\glv)$ is defined, in these works, to be the group of bijections $f:\glv \to \glv$ such that $f, f^{-1}$ take collinear points to collinear points. It follows from the proof that all such $f, f^{-1}$ also  take lines to lines. Another difference is that in the case $m= 2 \ell$, if  $f \in $ Aut$(\glv)$ is not in im$(\rho)$ then the proofs in these works establish that $f$ is induced by a \emph{correlation}, i.e., there is a semilinear isomorphism $\tilde \theta: V^* \to V$ such that $f$ is (the restriction to $\glv$ of) the projectivization of $\wedge^{\ell} \tilde \theta \circ  H_{\ell}$ . We recall the maps $H_{\ell}, \; \theta$ and  $\ast_{\ell}$ defined in Proposition~\ref{Hodge_prop} and its proof. Since $\theta^{-1} \circ \tilde \theta \in \gL(V)$, the full group Aut$(\glv)$  (when $m = 2 \ell$) is the subgroup of $\PgL (\wlv)$ generated by im$(\rho)$ and $\ast_{\ell}$, the structure of which is given in Corollary~\ref{Hodge_crlry}.

\subsection{Automorphism group of $C(\ell,m)$} \label{aut_clm}
Here, and in the remainder of this subsection, we take $\bK$ to be the finite field $F=\Fq$ with $q$ elements. In this case $\Glm$ is a finite set.  We will denote the  number $ \binom{m}{\ell}$  as $k$, and the cardinality $|\Glm|$ as $n$. The projective system $\mP  = \Glm \subset \bP(\wl)$ gives rise to the $[n,k]_q$ Grassmann code $\mC = C(\ell,m)$. Let $\{P_1, \cdots, P_n\}$ denote representatives in $\wl$ of the  $n$ points of $\Glm$ taken in some fixed order. The code $\mC$ is given by
\[ 
\mC = \{ ( \omega(P_1), \cdots,\omega(P_n) )  :  \omega \in \wld \}
\]
The $k \times n$ matrix $M$ whose $(i,j)$-th entry is the $i$-th Pl\"{u}cker coordinate of the $j$-th point of $\mP$ is a generator matrix for $\mC = C(\ell,m)$. This requires picking  an ordering of the $k$ elements of $\mathcal I_{\ell,m}$. Now Definition~\ref{aut_def} of the automorphism group gives:
\begin{eqnarray} \label{eq:aut_glm}
{\rm Aut}(\mP) &=& \{ g \in  \PgL ( \wl) \; : \;\; g(\Glm) = \Glm \}, \text{ and} \\ 
\nonumber
{\rm Aut}(\mC) &\simeq& \pi^{-1}({\rm Aut}(\mP)),  
\end{eqnarray}
where $\pi: \gL( \wl) \to  \PgL( \wl)$ is the canonical  homomorphism. 

\begin{theorem}  \label{AutP} 
For the projective system $\mP   = \Glm \subset \bP(\wl)$, we have:   
\[ \emph{Aut}(\mP) \simeq \begin{cases}  \PgL(m,F)  & \mbox{ if } \; m \neq 2 \ell, \\
\PgL(m,F)  \rtimes_{-t} \mathbb Z / 2 \mathbb Z  & \mbox{ if } \; m = 2 \ell. \end{cases} \]
\end{theorem} 

This result is a consequence of Theorem \ref{chow_thm} (see also Westwick \cite{Westwick}): The group $\Aut(\mP)$ according to \eqref{eq:aut_glm} consists of those elements of $\PgL (\wl)$ which take $\Glm$ to itself. Since elements of $\PgL (\wl)$  preserve the set of lines in $\Glm$,  the result follows from Theorem \ref{chow_thm}.
\qed \\

In the remaining part of this subsection we give a explicit description of the group Aut$(\mC) = \pi^{-1}($Aut$(\mP))$. This is the subgroup of $\gL(\wl)$ generated by im$(\hat \rho)$, the scalar matrices $F^{\times}$, and  if $m = 2 \ell$  also $\ast_{\ell}$. We need some definitions.

\begin{definition} 
{\rm 
Let $\mG$ denote the subgroup of $\GL(\wl)$ generated by $\hat \rho(\GL(m,F))$ and the group $F^{\times}$ of scalar matrices.  Let $\lambda = (q-1,\ell)$ be the GCD of $q-1$ and $\ell$, and let $\lambda' = (q-1)/\lambda$. Also let $\mu_{\lambda}$ and $\mu_{\lambda'}$ denote, respectively, the group of $\lambda$-th and $\lambda'$-th roots of unity in $F^{\times}$ (identified with the corresponding scalar matrices in $\GL(m,F)$).
}
\end{definition}

If $A\in \mu_{\lambda}$, then clearly $\wedge^{\ell}\!A$ is the identity matrix. Thus we have  a natural epimorphism
\[ 
\varrho: (\GL(m,F)/\mu_{\lambda}) \times F^{\times} \to \mG \quad \text{ given by } \quad (A,c)  \mapsto c \, (\wedge^{\ell}\!A). 
 \]

If $(A,d) \in$ ker$(\varrho)$ then $d \,(\wedge^{\ell} A)$ is the identity element of $\GL(\wl)$, which in turn implies that $[A] \in \PGL (m,F)$ is in ker$(\rho)$. Since ker$(\rho)$ is trivial, we obtain $A = c I_m$ for some $c \in F^{\times}/\mu_{\lambda}$. Since $d (\wedge^{\ell}A)$ is identity, we get $d = c^{-\ell}$. Thus 
\beq
\label{eq:K_def}
{\rm ker}(\varrho)=  \{ (c I_m,c^{-\ell})  : c \in F^{\times}/\mu_{\lambda} \} \subset (\GL(m,F)/\mu_{\lambda}) \times F^{\times} . 
\eeq
Therefore, upon letting $K$ denote $\ker \varrho$, we obtain an isomorphism: 
\beq \label{eq:mGK} 
\mG \simeq ((\GL(m,F)/\mu_{\lambda}) \times F^{\times})/K . 
\eeq

We now describe Aut$(\mC)$ in terms of $\mG$. The group Aut$(\mC) \subset \gL(\wl)$ is generated by $\mG$, Aut$(F)$ (acting entrywise on on $\GL(\wl)$), and 
also $\ast_{\ell}$ if $m  = 2 \ell$.  Writing $g \in \mG$ as $c(\wedge^{\ell}\!A)$, we see that $\sigma(g) = \sigma(c)  \wedge^{\ell}\!\sigma(A)$ is in $\mG$ for all $\sigma \in \Aut(F)$; hence we have an injection $\Aut(F) \hookrightarrow \Aut(\mG)$. In the case $m = 2 \ell$, we see that  $\ast_{\ell} \circ g \circ \ast_{\ell}^{-1} = c \, {\rm det}(A) (\wedge^{\ell}\!A^{-t})$ is in $\mG$. Since $\ast_{\ell}^2$ differs from $I_{\GL(\wl)}$ by a sign factor of  $(-1)^{\ell}$, it is more convenient to work with $\tilde \ast_{\ell} = (\wedge^{\ell}\kappa) \circ \ast_{\ell}$, where $\kappa \in \GL(m,F)$ is the matrix with ones on the antidiagonal and zeros elsewhere, i.e it is the matrix of the linear transformation that sends $e_i \mapsto e_{m-i+1}$. Using $\kappa = \kappa^{-1} = \kappa^t$ and det$(\kappa) = (-1)^{\ell (2 \ell-1)}$, we get $\tilde \ast_{\ell}^2 = (\wedge^{\ell}\kappa  \circ \ast_{\ell})^2 =\wedge^{\ell}(\kappa \kappa^{-t}) \,{\rm det}(\kappa) \ast_{\ell}^2 = I_{\GL(\wl)}$. Therefore we have an injection of the group  $\mathbb Z / 2 \mathbb Z$ generated by $\tilde \ast_{\ell}$ into Aut$(\mG)$. We also note that $\tilde \ast_{\ell} \notin \mG$, for otherwise  $\wedge^{\ell} [f] \mapsto \ast_{\ell} \circ \wedge^{\ell} [f] \circ  \ast_{\ell}^{-1}$ would be an inner automorphism of $\im (\rho)$, which is not true (Proposition \ref{Hodge_prop}). The automorphisms of $\mG$ corresponding to Aut$(F)$ and $\tilde \ast_{\ell}$ commute, and hence we have an injection Aut$(F) \times \mathbb Z / 2 \mathbb Z  \hookrightarrow $ Aut$(\mG)$. In summary:
\beq \label{eq:mCmG}
{\rm Aut}(\mC)= \begin{cases}  \mG  \rtimes  {\rm Aut}(F)   & \mbox{ if }   m \neq 2 \ell, \\
                               \mG \rtimes  ({\rm Aut}(F)  \times \mathbb Z / 2 \mathbb Z) & \mbox{ if }  m=2 \ell. 
\end{cases} 
\eeq
 In order to clarify the structure of $\mG$ as a central extension of $\PGL (m,F)$ we consider the following exact sequences:
\begin{eqnarray} \label{eq:exct1}
 1 \to &\mu_{\lambda'}& \xrightarrow{\imath_1} \GL(m,F)/\mu_{\lambda} \xrightarrow{\jmath_1} \PGL(m,F) \to 1 \\  \label{eq:exct2}
 1 \to &\mu_{\lambda'}& \xrightarrow{\imath_2} F^{\times}  \xrightarrow{\jmath_2} \mu_{\lambda} \to 1
\end{eqnarray}
\beq \label{eq:exct3}
1 \to \mu_{\lambda'} \times \mu_{\lambda'} 
\xrightarrow{\imath_1 \times \imath_2} (\GL(m,F)/\mu_{\lambda}) \times F^{\times}  \xrightarrow{\jmath_1 \times \jmath_2} \PGL (m,F) \times \mu_{\lambda} \to 1
 \eeq
where the map $\imath_1$ is given by $d \mapsto d^{1/\lambda} I_m$ (the map $d \mapsto d^{1/\lambda}$  is the inverse of the isomorphism
$F^{\times}/\mu_{\lambda} \to \mu_{\lambda'}$ given by $c \mapsto c^{\lambda}$), whereas the map $\imath_2$ is given by $d \mapsto d^{\ell/\lambda}$ (note that $\imath_2$ is injective because $(\ell/\lambda,\lambda')=1$). The map $\jmath_1$ is the projectivization map induced from $\GL(m,F) \to \PGL(m,F)$. The map $\jmath_2$ takes $c \mapsto c^{\lambda'}$.  The exact sequence \eqref{eq:exct3} is obtained from \eqref{eq:exct1}-\eqref{eq:exct2} by taking products in each factor. We now consider the subgroup $K \subset  (\GL(m,F)/\mu_{\lambda})\times F^{\times}$, and the subgroup $K'  \subset \mu_{\lambda'} \times \mu_{\lambda'}$ defined by:
\[ 
K'=\{ (d,d^{-1})  : d \in \mu_{\lambda'}\} = \{(c^{\lambda},c^{-\lambda})  : c \in F^{\times}/\mu_{\lambda}\}
 \subset \mu_{\lambda'} \times \mu_{\lambda'}.
\]

We observe that the map $\imath_1 \times \imath_2$ of \eqref{eq:exct3} carries $K'$ isomorphically to $K$. Thus passing to the quotients $(\mu_{\lambda'} \times \mu_{\lambda'})/K'  \simeq \mu_{\lambda'}$  and $((\GL(m,F)/\mu_{\lambda}) \times F^{\times})/K \simeq \mG$ in \eqref{eq:exct3}, we get the following exact sequence:
\beq \label{eq:exct4}
 1 \to \mu_{\lambda'} \xrightarrow{[\imath_1 \times \imath_2]} \mG  \xrightarrow{[\jmath_1 \times \jmath_2]} \PGL (m,F) \times \mu_{\lambda} \to 1
\eeq
Thus the exact sequence  \eqref{eq:exct4} expresses $\mG$  as a central extension by $\mu_{\lambda'}$ of $\PGL(m,F) \times \mu_{\lambda}$. The process of arriving at  \eqref{eq:exct4} from \eqref{eq:exct1}-\eqref{eq:exct2}  can be understood in terms of group cohomology. Using the correspondence between central $\mu_{\lambda'}$-extensions of a group $G$ and elements of  $H^2(G,\mu_{\lambda'})$ (with $G$ acting trivially on $\mu_{\lambda'}$), let $\alpha \in  Z^2(\PGL(m,F) ,\mu_{\lambda'})$ and  $\beta \in Z^2(\mu_{\lambda},\mu_{\lambda'})$ be cocycles representing \eqref{eq:exct1} and \eqref{eq:exct2}.
If $p_1$ and $p_2$ denote the projection homomorphisms from $\PGL(m,F)  \times \mu_{\lambda}$ to $\PGL(m,F)$ and   $\mu_{\lambda}$ respectively, then
the $\mu_{\lambda'}$-extension of $\PGL(m,F)  \times \mu_{\lambda}$ that corresponds to the cocycle  $p_1^*(\alpha) +p_2^*(\beta)$
is given by the amalgamated central product of the extensions corresponding to $\alpha$ and $\beta$, i.e.,   $((\GL(m,F)/\mu_{\lambda}) \times F^{\times})/K$. We now have a complete description of $\mG$ and (using \eqref{eq:mCmG}) of Aut$(\mC)$ .
 
\begin{theorem} \label{AutC}
Let $\lambda=(q-1,\ell)$, $\lambda' = (q-1)/\lambda$, and let $\mu_{\lambda}$ and $\mu_{\lambda'}$ denote, respectively, the group of $\lambda$-th and $\lambda'$-th roots of unity in $F^{\times}$. Let $p_1, p_2$ denote the projections of $\PGL(m,F)  \times \mu_{\lambda}$ onto its factors. The group $\mG$ is a central extension of $(\PGL(m,F) \times \mu_{\lambda})$ by $\mu_{\lambda'}$, corresponding to the class  $[p_1^*(\alpha)] +[p_2^*(\beta)] \in  H^{2}((\PGL(m,F)  \times \mu_{\lambda}),\mu_{\lambda'})$, where $[\alpha] \in  H^2(\PGL(m,F),\mu_{\lambda'})$ and  $[\beta] \in H^2(\mu_{\lambda},\mu_{\lambda'})$ are classes representing the $\mu_{\lambda'}$-extensions $ \GL(m,F)/\mu_{\lambda}$ and $F^{\times}$ appearing in \eqref{eq:exct1} and \eqref{eq:exct2}. 
\end{theorem}

In the special case when  $(\lambda,\lambda')=1$, the extension $F^{\times}$ splits as $F^{\times} = \mu_{\lambda} \times \mu_{\lambda'}$. The class $[\beta]= 0$, and hence  $[p_1^*(\alpha)] +[p_2^*(\beta)] = [p_1^*(\alpha)]$. Thus $\mG \simeq (\GL(m,F)/ \mu_{\lambda})  \times \mu_{\lambda}$.  In particular if 
$\lambda = (q-1,\ell)= 1$, the group  $\mu_{\lambda}$ is trivial, and so it follows that  $\mG \simeq \GL(m,F)$. Similarly, if $\ell \equiv 0\imod{q-1}$, then $\mu_{\lambda} = F^{\times}$ and hence $\mG \simeq \PGL(m,F) \times F^{\times}$. \\

\noi
{\bf Remark:} The monomial automorphism group of the Grassmann $\mC(\ell,m)$, viz.,  $\MAut(\mC(\ell,m)) = {\rm Aut}(\mC(\ell,m)) \cap \GL(\wl)$ is $\mG$ if $m \neq 2 \ell$ and $\mG \rtimes  \mathbb Z / 2 \mathbb Z$ if $m  = 2 \ell$, where the semidirect product is as in \eqref{eq:mCmG}.

\section{Affine Grassmann codes} \label{sec4}

We continue to use the notations of the previous section. In particular, integers $\ell, m$ are kept fixed throughout and it may be tacitly assumed that $1< \ell < m$. To begin with, suppose the ground field $\bK$ is the finite field $F=\Fq$ with $q$ elements.  We recall $\mP = \glm \subset \bP(\wl)$ is the projective system defining the Grassmann code $C(\ell,m)$. Let $\mH_0$ denote the coordinate hyperplane of $\bP(\wl)$ given by the vanishing of the Pl\"{u}cker coordinate $p_{I_0}$ where $I_0:=(m-\ell+1, m-\ell+2, \dots,m)$. The intersection of $\mH_0$ with $\glm$ is a Schubert divisor of $\glm$. It will be denoted as $\Omega$.  The complement of $\Omega$, i.e., $\glm \setminus\!\mH_0$ will be denoted $W_0$. It is the big cell (of dimension $\ell(m-\ell)$) of $\glm$ in its standard cell structure. Let  $V_{m-\ell}$ denote the span of the first $m-\ell$ standard basic vectors  of $F^m$. We can also describe $\omg$ and $W_0$ as:
\beqn \label{eq:W_0_def}
\Omega& =& \{ \gamma \in \glm  : {\rm dim}(\gamma \cap V_{m-\ell}) > 0\} \text{ and } \\ \nonumber
 W_0 &=& \{ \gamma \in \glm  : {\rm dim }(\gamma \cap V_{m-\ell}) = 0\}. 
\eeqn

We consider the projective system $\mP^{\bA} = W_0 \subset \bP(\wl)$. It is known that the minimum distance of the Grassmann code $C(\ell,m)$ is $|W_0| = q^{\ell(m-\ell)}$ (see \cite{Nogin}). In other words a hyperplane of $\bP(\wl)$ can intersect $\glm$ in at most $|\glm|-|W_0| = O(q^{\ell(m-\ell)-1})$ points. As a consequence, we see that the projective system $\mP^{\bA} \subset \bP(\wl)$ is nondegenerate. The affine Grassmann code $C^{\bA}(\ell,m)$ is the code associated with this projective system. The length and dimension of this code are $n = q^{\ell(m-\ell)}$ and $k = \binom{m}{\ell}$ respectively. By construction, it is obtained from the Grassmann code $C(\ell,m)$ by puncturing on the coordinates corresponding to the points of $\mP \!\setminus\! \mP^{\bA}$. A $k \times n$ generator matrix $M^{\bA}$ for this code can be constructed as follows: We pick some ordering of the $k$ Pl\"{u}cker coordinates $\{p_I  :  I \in \mathcal I(\ell,m)\}$, such that the first of these is $p_{I_0}$ (as defined above). We pick some ordering of the $n$ points of $W_0$. The $j$-th column of the $k \times n$ generator matrix $M^{\bA}$ is the coordinate vector $\{p_I/p_{I_0}  : I \in \mathcal I(\ell,m)\}$ (in the chosen order) of the $j$-th point of $W_0$. The code $C^{\bA} (\ell,m)$ is also an example of an algebraic-geometric code $C(X,\mathcal L; \mP)$ (see \cite[Chapter 3.1]{T-V}) constructed by evaluating the global sections $H^0(X,\mathcal L(D))$ of a line bundle $\mathcal L(D)$ associated with a divisor $D$ of a smooth projective variety $X$ over $F$. In this construction the global sections are evaluated on some subset $\mP =\{P_1, \dots,P_n\}$ of $F$-rational points of $X$ (after choosing an isomorphism of $\mathcal L_{P_i}$ with $F$) disjoint from the support of $D$. For the affine Grassmann code, the triple $(X, D, \mP)$ is $(\glm, \omg, W_0)$. To see this we note that the vector space generated by $\{p_I/p_{I_0}  :  I \in \mathcal I(\ell,m)\}$ (viewed as functions on $W_0$) is just the space $L(\omg): = \{0\} \cup \{ f \in \bF_q(\glm)^{\times}  : {\rm div}(f) + \omg \geq 0\}$. Therefore the code $C(X,\mathcal L; \mP)$ is monomially equivalent to $C^{\bA} (\ell,m)$. The affine Grassmann code was introduced in a slightly different but equivalent formulation by Beelen, Ghorpade and H{\o}holdt \cite{BGH1}. For any $\gamma \in W_0$, there is a unique $\ell \times (m -\ell)$ matrix $A$ such that the rows of the $\ell \times m$ matrix $(A \,|\,  I_{\ell} )$  (where $I_{\ell}$ is the $\ell \times \ell$ identity matrix) form a basis for $\gamma$ (\cite[p. 193]{GH}). The Pl\"{u}cker coordinates $\{p_I/p_{I_0}  :  I \in \mathcal I(\ell,m)\}$ of a point $\gamma \in W_0$ were interpreted in  \cite{BGH1}, as the set of all $r \times r$ minors (for $0 \leq r \leq \ell$) of the matrix $A$ described above. In particular, any generator matrix for the code constructed in  \cite{BGH1}, differs from $M^{\bA}$ only by  a row transformation and  a column permutation. 

In this section we  determine the automorphism group $\Aut(\mP^{\bA})$ of the projective system, or equivalently  $\Aut(C^{\bA}(\ell,m))/F^{\times}$ (the automorphism group of the code modulo dilations). The group $\Aut(C^{\bA}(\ell,m))$ itself is a central $F^{\times}$-extension of $\Aut(\mP^{\bA})$. Of more interest in the case of affine Grassmann codes, is the permutation automorphism group, $\PAut(C^{\bA}(\ell,m))$. This is the subgroup of  the permutation group on the columns of $M^{\bA}$, consisting of permutations which preserve the row space of $M^{\bA}$. In \cite{BGH2}, a subgroup of $\PAut(C^{\bA}(\ell,m))$ was identified, and it was remarked that the full group could be larger and it was shown that this subgroup excludes an element (in fact, an involution) of $\PAut(C^{\bA}(\ell,m))$ when $\ell = m-\ell$. We show that the subgroup determined in \cite{BGH2} is, in fact, the full group $\PAut(C^{\bA}(\ell,m))$ when $m\ne 2\ell$, and that the excluded involution is essentially the only missing ingredient when $m=2\ell$. 

\subsection{Automorphisms of the big cell of the Grassmannian}
In this subsection, the ground field $\bK$ is arbitrary. Let $V = \bK^m$, and  let $V_{m-\ell} \subset V$ be the subspace spanned by the first $m-\ell$ standard basic vectors $\{e_1, \dots, e_{m-\ell}\}$. We recall that $W_0 \subset \glv$ consists of those $\gamma$ which are complementary to $V_{m-\ell}$. We define:
\[ 
{\rm Aut}(W_0) = \{g \in \PgL (\wlv)  :  g(W_0) = W_0\}.
\]
We introduce a decomposition of $\glv$ as $W_0 \cup W_1 \cup \dots  \cup W_{\ell}$ where: 
\beq \label{eq:W_i_def}
W_i: = \{\gamma \in \glv  :  {\rm dim}(\gamma \cap V_{m-\ell}) = i\}, \quad  \mbox{for} \;\, 0 \leq i \leq \ell .
\eeq
The following elementary lemma follows from the fact that $\glv \subset \bP(\wlv)$ is cut out by quadratic polynomials.
\begin{lemma} \label{line_glv}
If a line $L \subset \bP(\wlv)$ intersects $\glv$ in two distinct points $\gamma, \gamma'$, then either $L \subset \glv$ or $L \cap \glv = \{\gamma,\gamma'\}$.  
\end{lemma}

\begin{theorem} \label{aut_w0}
Viewed as subgroups of $\PgL (\wlv)$, we have \emph{Aut}$(W_0) \subset $\emph{Aut}$(\glv)$. Moreover, every $g \in\,$\emph{Aut}$(W_0)$ also satisfies 
$g(W_i) = W_i$ for $0 \leq i \leq \ell$.
\end{theorem}

\noi
{\bf Remark:} The theorem is motivated by the special case when $\bK$ is algebraically closed. In this case for any $g \in \PGL(\wlv)$ with $g(W_0)=W_0$,  the intersection $g(\glv) \cap \glv$ is a closed subvariety of $\glv$ containing $W_0$. Since $\glv$ is the Zariski closure of $W_0$, it follows that $g(\glv) \cap \glv = \glv$.  For arbitrary fields this reasoning has to be replaced by an argument such as the following.

\begin{proof}  The first assertion follows from the second, because $\glv$ is the union $W_0 \cup \dots \cup W_{\ell}$. We prove the second assertion by induction on $i$. The  base case $i=0$ is true by hypothesis. We assume inductively that $g(W_i) = W_i$ for $0 \leq i \leq r-1$, and establish the result for $i=r$.  For any  $\gamma \in W_r$, we pick an $(\ell-1)$-dimensional subspace $\beta \subset \gamma$ such that $\dim (\beta \cap V_{m-\ell}) = r-1$, and pick $u_1 \in V_{m-\ell}$ satisfying  $\gamma = \beta \wedge u_1$. We pick any $u_2 \notin (\beta+V_{m-\ell})$ and consider $\gamma' = \beta \wedge (u_1+u_2)$. The line $L$ joining $\gamma$ and $\gamma'$ has exactly one point in $L \cap W_{r}$ namely $\gamma$, the remaining  points of $L$ are in $W_{r-1}$. Now let $h$ be any one of the two maps $g, g^{-1}$. Since $h \in \PgL (V)$, the image $h(L)$ is a line in $\bP(\wlv)$, and since $h(W_{r-1}) = W_{r-1}$ we see that $h(L) \!\setminus\! \glv$ has at most one point. Therefore, by Lemma \ref{line_glv}, we have $h(L) \subset \glv$, i.e., $h(\gamma) \in \glv$. It remains to prove that $h(\gamma) \in W_{r}$. To this end we make two observations about $h(L)$:  the first is that $|h(L) \cap W_{r-1}| >1$, because as mentioned above there is  atmost one point of $h(L)$ which is not in $W_{r-1}$. The second observation is that $h(L)$  has atleast one point not in $W_{r-1}$, because if $h(L) \subset W_{r-1}$, then the fact that $h^{-1}$ takes $W_{r-1}$ to itself would imply that  $L \subset W_{r-1}$ which is not true. Together these observations imply that $h(\gamma) \notin W_{r-1}$ and all the other points of $h(L)$ are in $W_{r-1}$. Since $h(L)$ is a line on $\glv$, we can express it as $h(L) =\pi_{\beta'}^{\delta}$ for some $\beta' \in \gkv$ and  $\delta \in \gmv$. Any $\gamma' \in h(L)$ other than $h(\gamma)$  is in $W_{r-1}$, and hence  $\dim (\beta' \cap V_{m-\ell})$ is  $r-2$ or $r-1$ and $\dim (\delta \cap V_{m-\ell})$ is  $r-1$ or $r$. Of these four cases, in the case  $\dim (\beta' \cap V_{m-\ell})=r-1$ and $\dim (\delta \cap V_{m-\ell})=r$, the line $h(L)$ has exactly one point in $W_r$ which must be $h(\gamma)$ as was to be shown. In the other three cases, either  $h(L) \subset W_{r-1}$ or $|h(L) \cap W_{r-1}| =1$ which have been observed above to be false. Since $\gamma \in W_r$ was arbitrary, we have shown that $h(W_r) \subset W_r$. Since $h$ was allowed to be either of the maps $g, g^{-1}$, we find $g(W_r) = W_r$. 
\end{proof} 

We need some more notation before we can describe the subgroup $\Aut(W_0)$ of  $\Aut(\glv)$.  Let $P_{m-\ell,\ell}$ be the maximal parabolic subgroup of $\GL(V)$ consisting of matrices that carry $V_{m-\ell}$ to itself. (i.e.,  matrices whose $\ell \times m-\ell$ submatrix on the last $\ell$ rows and first $m-\ell$ columns is zero).
Let  $\hat P_{m-\ell,\ell} := P_{m-\ell,\ell} \rtimes {\rm Aut}(\bK)$ denote the subgroup of $\gL(V)$ consisting of transformations that carry $V_{m-\ell}$ to itself, and let $\hat P_{m-\ell,\ell}/\bK^{\times}$ be the corresponding subgroup of $\PgL (V)$. In the case  when $m = 2 \ell$, let $- \tilde t$ denote the automorphism $g \mapsto \kappa g^{-t} \kappa^{-1} = (\kappa g \kappa^{-1})^{-t}$ of $\GL(V)$ (where $\kappa$ was defined before \eqref{eq:mCmG}). There is an induced automorphism of  $PGL(V)$ and $\PgL  (V)$, again denoted by  $-\tilde t$. For $g \in P_{\ell,\ell}$, the matrix $\kappa g \kappa^{-1}$ has a zero for the submatrix on first $\ell$ rows and last $\ell$ columns. Therefore $(\kappa g \kappa^{-1})^{-t}$ is again in $ P_{\ell,\ell}$, and hence $-\tilde t$ is an  automorphism of $P_{\ell,\ell}/\bK^{\times}$ and $\hat P_{\ell,\ell}/\bK^{\times}$.

\begin{definition}  \label{Pll2_def}
{\rm 
Let $(P_{\ell,\ell}/ \bK^{\times}) \rtimes_{-\tilde t} \mathbb Z/ 2 \mathbb Z$ and $(\hat P_{\ell,\ell}/ \bK^{\times}) \rtimes_{-\tilde t} \mathbb Z/ 2 \mathbb Z$ denote
 the semidirect product of $P_{\ell,\ell}/\bK^{\times}$ and $\hat P_{\ell,\ell}/\bK^{\times}$ respectively, with $\mathbb Z/ 2 \mathbb Z$ where the generator of $\mathbb Z/ 2 \mathbb Z$ acts by the automorphism $-\tilde t$. 
}
\end{definition}

Continuing further with the case $m = 2 \ell$, we note that  $\tilde \ast_{\ell} \in \gL(\wlv)$ permutes (upto sign) the basic vectors $\{e_I :  I \in \mathcal I(\ell,2 \ell) \}$ of $\wlv$, and in particular sends $e_{I_0} \mapsto (-1)^{\ell (\ell+1)/2}  e_{I_0}$, where as before $I_0 = (\ell+1, \dots, 2 \ell)$. Therefore,  $\tilde \ast_{\ell} \in \PgL(\wlv)$ preserves the hyperplane $\mH_0 \subset \bP(\wlv)$. Since $\tilde \ast_{\ell} \in$ Aut$(\glv)$,  it  takes $W_0 = \glv \setminus \mH_0$ to itself, i.e.,  $\tilde \ast_{\ell} \in$ Aut$(W_0)$.  Now, $\Aut(\glv)$ is the subgroup of $\PgL (\wlv)$ generated by $\im (\rho)$ and $\ast_{\ell}$, where we may replace $\ast_{\ell}$ by $\tilde \ast_{\ell}$ because $( \ast_{\ell})^{-1} \tilde \ast_{\ell} = \wedge^{\ell} \kappa  \in \im (\rho)$. Writing an element of Aut$(\glv)$ as $\rho(f) \circ \tilde \ast_{\ell}^{i}$ with $i \in\{0,1\}$, it follows that  $\Aut(W_0)$ is generated by $\Aut(W_0) \cap \im (\rho)$ and $\tilde \ast_{\ell}$. Since, $\ast_{\ell}\circ \wedge^{\ell} f \circ  \ast_{\ell}^{-1} = \wedge^{\ell} f^{-t}$ in $\PgL(\wlv)$ (Proposition \ref{Hodge_prop}), we get $\tilde \ast_{\ell}\circ \wedge^{\ell} f \circ  {\tilde \ast_{\ell}}^{-1} = \wedge^{\ell}(\kappa  f^{- t} \kappa^{-1})$. We note that $\tilde \ast_{\ell} \notin$ im$(\rho)$, hence $\tilde \ast_{\ell} \notin 
 \rho( \hat P_{\ell,\ell}/\bK^{\times})$. We denote the subgroup of $\PgL (\wlv)$ generated by 
$\tilde \ast_{\ell}$ and $\rho( \hat P_{\ell,\ell}/\bK^{\times})$ by $\rho( \hat P_{\ell,\ell}/ \bK^{\times})\rtimes_{\tilde \ast_{\ell}} \mathbb Z/ 2 \mathbb Z$.
Writing:  
\beqs
 (\hat P_{\ell,\ell}/ \bK^{\times})\!\rtimes_{-\tilde t} \mathbb Z/ 2 \mathbb Z \!\!\!\!&=&\!\!\!\! \langle \hat P_{\ell,\ell}/ \bK^{\times}, \epsilon \, |\, \epsilon^2=I_{\PgL(m,\bK)}, \epsilon f \epsilon^{-1} = \kappa  f^{-t} \kappa^{-1} \rangle, \; \mbox{ and }\\
\rho(\hat P_{\ell,\ell}/ \bK^{\times})\!\rtimes_{\tilde \ast_{\ell}} \mathbb Z/ 2 \mathbb Z  \!\!\!\!&=&\!\!\!\!
\langle \rho( \hat P_{\ell,\ell}/ \bK^{\times}), \tilde \ast_{\ell} \, |\, \tilde \ast_{\ell}^2 = I_{\PgL(\wlk)},  \tilde \ast_{\ell} \wedge^{\ell}\!\!f \tilde \ast_{\ell}^{-1} = \wedge^{\ell} \!(\kappa f^{-t} \kappa^{-1}) \rangle,
\eeqs
 it follows that the map from $(\hat P_{\ell,\ell}/ \bK^{\times}) \rtimes_{-\tilde t} \mathbb Z/ 2 \mathbb Z$ to $\rho( \hat P_{\ell,\ell}/ \bK^{\times})\rtimes_{\tilde \ast_{\ell}} \mathbb Z/ 2 \mathbb Z$ that sends  $f \mapsto \wedge^{\ell}f$ and $\epsilon \mapsto \tilde \ast_{\ell}$  establishes an isomorphism between these groups.

\begin{corollary}  \label{aut_w0_1}
For the big cell $W_0$ of the Grassmannian $G_{\ell}(V)$, we have
\[
{\rm Aut}(W_0) = \begin{cases} \rho( \hat P_{m-\ell,\ell}/\bK^{\times} )  \quad \simeq \; \hat P_{m-\ell,\ell}/\bK^{\times}  & \mbox{ if }  \; m \neq 2 \ell, \\
\rho( \hat P_{\ell,\ell}/ \bK^{\times})\rtimes_{\tilde \ast_{\ell}} \mathbb Z/ 2 \mathbb Z \quad  \simeq \; (\hat P_{\ell,\ell}/ \bK^{\times}) \rtimes_{-\tilde t} \mathbb Z/ 2 \mathbb Z & \mbox{ if } \;  m = 2 \ell. 
\end{cases}
\]
\end{corollary}

\begin{proof} 
Since any $g \in \hat P_{m-\ell,\ell}/\bK^{\times}$, by definition carries $V_{m-\ell}$ to itself, we see that $\rho(g)$  carries $W_0 = \{\gamma \in \glv  :  {\rm dim}(\gamma \cap V_{m-\ell}) = 0\}$  to itself.  Conversely, suppose $f \in \im (\rho)$ is in $\Aut(W_0)$.  Writing $f = \wedge^{\ell} g$ for some $g \in \PgL (V)$, we will prove that  $g(V_{m-\ell}) =  V_{m-\ell}$, and hence conclude that $f \in \hat P_{m-\ell,\ell}/\bK^{\times}$. To prove this suppose there is a $v_0 \in V_{m-\ell}$ with $g(v_0) \notin V_{m-\ell}$. We  pick a $\gamma \in G_{\ell}(V_{m-\ell}) = W_{\ell}$, such that $v_0 \in \gamma$. Since $f(\gamma)$, i.e., the image of the vector space $\gamma$ under $g$ contains $g(v_0) \notin V_{m-\ell}$, it follows that $f(\gamma) \notin W_{\ell}$, contradicting  the observation in Theorem \ref{aut_w0}, that $f$ carries each $W_i$ to itself bijectively, in particular $W_{\ell}$. In the case $m \neq 2 \ell$, since $\Aut(\glv)$ is $\im (\rho)$ and $\Aut(W_0) \subset \Aut(\glv)$, we conclude that $\Aut(W_0) =\rho(\hat P_{m-\ell,\ell}/\bK^{\times})$.  On the other hand, when $m = 2 \ell$, we have already observed that $\Aut(W_0)$ is generated by  $\tilde \ast_{\ell}$ and  $\Aut(W_0) \cap \im (\rho) = \rho(\hat P_{\ell,\ell}/ \bK^{\times})$. Hence $\Aut(W_0) = \rho( \hat P_{\ell,\ell}/ \bK^{\times})\rtimes_{\tilde \ast_{\ell}} \mathbb Z/ 2 \mathbb Z $.
\end{proof}

\subsection{Permutation Automorphisms  of $C^{\bA}(\ell,m)$}
To begin with, let us consider the semilinear as well as the monomial automorphisms of the affine Grassmann code $C^{\bA}(\ell,m)$.
By definition, the automorphism group of the projective system $\mP^{\bA} = W_0 \subset \bP(\wl)$ is the group $\Aut(W_0) \subset \PgL(\wl)$ determined in Corollary \ref{aut_w0_1}. The group $\Aut(C^{\bA}(\ell,m))$ is the subgroup $\pi^{-1}(\Aut(W_0))$, where $\pi:\gL(\wl) \to \PgL(\wl)$ is the quotient homomorphism.
In other words  $\Aut(C^{\bA}(\ell,m))$ is the subgroup of $\gL(\wl)$ generated by $\hat \rho( \hat P_{m-\ell,\ell})$ and  $F^{\times}$ and also $\tilde \ast_{\ell}$ if $m = 2 \ell$. Let $\mG^A$ denote the subgroup of $\GL(\wl)$ generated by $\hat \rho( P_{m-\ell,\ell})$ and $F^{\times}$. We recall from Section \ref{aut_clm} the epimorphism $\varrho: (\GL(m,F)/\mu_{\lambda}) \times F^{\times} \to \mG$. Let $\varrho_1$ denote the restriction of $\varrho$ to  $(P_{m-\ell,\ell}/\mu_{\lambda}) \times F^{\times}$. The image of $\varrho_1$ is $\mG^A$ and since $K = $ ker$(\varrho)$  (see \eqref{eq:K_def})  is contained in  $(P_{m-\ell,\ell}/\mu_{\lambda}) \times F^{\times}$, we have  ker$(\varrho_1) = K$ as well. Therefore, $\varrho_1$ induces  an isomorphism
\beq 
\label{eq:mGAK} 
\mG^A \simeq ( (P_{m-\ell,\ell}/\mu_{\lambda}) \times F^{\times})/K .
\eeq

The group MAut$(C^{\bA}(\ell,m)) =$ Aut$(C^{\bA}(\ell,m)) \cap GL(\wl)$ is $\mG^A$ if $m \neq 2 \ell$ and generated by $\mG^A$ and $\tilde \ast_{\ell}$ if $m  = 2 \ell$. The group Aut$(C^{\bA}(\ell,m))$ is generated by MAut$(C^{\bA}(\ell,m))$ and Aut$(F)$. Since the automorphisms of Aut$(\mG)$ induced by Aut$(F)$ and $\tilde \ast_{\ell}$ (see the discussion preceding \eqref{eq:mCmG}) clearly take $\mG^A$ to itself we have injections  Aut$(F) \hookrightarrow $ Aut$(\mG)$, and in case $m = 2 \ell$,  Aut$(F) \times \mathbb Z / 2 \mathbb Z  \hookrightarrow $ Aut$(\mG)$. 
In summary:

\begin{theorem} \label{thm:AutAff} 
The automorphism group of the affine Grassmann code $C^{\bA}(\ell,m)$ is given by 
\beqs \label{eq:mCmGA}
{\rm Aut}(C^{\bA}(\ell,m))= \begin{cases}  \mG^A  \rtimes  {\rm Aut}(F)   & \mbox{ if }   m \neq 2 \ell, \\
                               \mG^A \rtimes  ({\rm Aut}(F)  \times \mathbb Z / 2 \mathbb Z) & \mbox{ if }  m=2 \ell.  \end{cases}  
\eeqs
whereas the monomial automorphism group of $C^{\bA}(\ell,m)$ is given by 
\beqs
{\rm MAut}(C^{\bA}(\ell,m)) = \begin{cases}  \mG^A  & \mbox{ if }   m \neq 2 \ell, \\
                               \mG^A \rtimes  \mathbb Z / 2 \mathbb Z & \mbox{ if }  m=2 \ell.  \end{cases}
\eeqs
\end{theorem}

Since we are mainly interested in the permutation automorphism group of  $C^{\bA}(\ell,m)$, we do not further  discuss the structure of $\Aut(C^{\bA}(\ell,m))$ as a central $F^{\times}$-extension of Aut$(W_0)$.  We now describe $\PAut(C^{\bA}(\ell,m))$.

\begin{theorem}
\label{thm:PAutAff} 
The permutation automorphism group of the affine Grassmann code $C^{\bA}(\ell,m)$ is given (up to isomorphism) by 
\[  {\rm PAut}(C^{\bA}(\ell,m)) \simeq 
\begin{cases}   P_{m-\ell,\ell}/F^{\times}  & \mbox{ if } \; m \neq 2 \ell, \\
 P_{\ell,\ell}/ F^{\times} \rtimes_{-\tilde t} \mathbb Z/ 2 \mathbb Z  & \mbox{ if } \; m = 2 \ell, \end{cases} 
\]
 where $P_{m-\ell,\ell}$ is the maximal parabolic subgroup of $\GL(V)$ of transformations preserving $V_{m-\ell}$ (the span of the first $m-\ell$ standard basic vectors of $F^m$), and where  $F^{\times} \subset P_{m-\ell,\ell}$ denotes the subgroup of scalar matrices, and $-\tilde t$ is the automorphism of $P_{\ell,\ell}$ defined in the discussion preceding Definition \ref{Pll2_def}.
\end{theorem}

\begin{proof}
$\PAut(C^{\bA}(\ell,m))$ is the subgroup of MAut$(C^{\bA}(\ell,m))$ consisting of those automorphisms which  preserve the set of vectors in $\wl$ having  Pl\"{u}cker coordinate $p_{I_0} = 1$, where $I_0:=(m-\ell+1, m-\ell+2, \dots,m)$.  We write a general element of $P_{m-\ell,\ell}$ as $\bbsm A & {\mathbf{u}}\\ 0 & B \besm$, where $A \in \GL(m-\ell,F)$, $B \in \GL(\ell,F)$ and  ${\mathbf{u}}$ is an $m-\ell \times \ell$ matrix.  We first consider the case $m \neq 2 \ell$. In this case MAut$(C^{\bA}(\ell,m)) = \mG^A$, and the subgroup PAut$(C^{\bA}(\ell,m))$ is the the subgroup of $\mG^A_1 \subset \mG^A$ consisting of elements which preserve the set of $\xi \in \wl$ with $p_{I_0}(\xi) = 1$. Let $H: = \varrho_1^{-1}(\mG^A_1)$. In other words:  
\beq
H :=  \left\{\left(\bbsm A & {\mathbf{u}}\\ 0 & B \besm \, {\rm mod }\, \mu_{\lambda}, \; d\right)  :  {\rm det}(B) =d^{-1} \right\} \subset P_{m-\ell,\ell}/ \mu_{\lambda}   \times F^{\times}.
\eeq
By definition of $H$ we see that $\varrho_1$ induces an isomorphism of $H/K$ with  $\mG^A_1 =$  PAut$(C^{\bA}(\ell,m))$ (where as before $K=$ ker$(\varrho) = $ker$(\varrho_1)$). Now, the epimorphism from $H$ to  $P_{m-\ell,\ell}/F^{\times}$ given by
\[   
(\bbsm A & {\mathbf{u}}\\ 0 & B \besm \, {\rm mod }\, \mu_{\lambda}, \; {\rm det}(B)^{-1})   \longmapsto 
\bbsm A & {\mathbf{u}}\\ 0 & B \besm \, {\rm mod }\, F^{\times} 
\]
has exactly $K$ for its kernel, thus proving that PAut$(C^{\bA}(\ell,m)) \simeq P_{m-\ell,\ell}/F^{\times}$. 

Next, we consider the case $m = 2 \ell$. Here, MAut$(C^{\bA}(\ell,m)) $ is the subgroup of $\GL(\wl)$ generated by $\mG^A$ and $\tilde \ast_{\ell}$:
\[ 
{\rm MAut}(C^{\bA}(\ell,m)) = \left\langle \,\mG^A, \tilde \ast_{\ell} \, |\,  \tilde \ast_{\ell}^2 = I_{\GL(\wl)},  \tilde \ast_{\ell}(c \,\wedge^{\ell}\!f)  \tilde \ast_{\ell}^{-1} = 
c \,{\rm det}(f) \, \wedge^{\ell}\!(\kappa f^{-t}\kappa^{-1}) \,\right\rangle.
\]
We consider the group:
\[ 
\tilde G:=\left\langle P_{\ell,\ell}/ \mu_{\lambda}   \times F^{\times}, \epsilon \,|\,  \epsilon^2 = I_{P_{\ell,\ell}/ \mu_{\lambda}   \times F^{\times}}, \,  \epsilon (f, c) \epsilon^{-1} = (\kappa f^{-t} \kappa^{-1}, {\rm det}(f) \, c )\right\rangle. 
\]
(where $\kappa f^{-t} \kappa^{-1}$mod $\mu_{\lambda}$ and det$(f)$ only depend on $f$ mod $\mu_{\lambda}$, because $-\tilde t$ preserves $\mu_{\lambda}$
and $\lambda$ divides $m = 2 \ell$).  The map $\varrho_1':\tilde G \to {\rm MAut}(C^{\bA}(\ell,m))$ that equals $\varrho_1$ on $P_{\ell,\ell}/ \mu_{\lambda}  \times F^{\times}$ and that sends $\epsilon \mapsto \tilde \ast_{\ell}$ is a well defined epimorphism, with ker$(\varrho_1') = K$ again. 
Since  $\tilde \ast_{\ell}$ sends  $e_{I_0} \mapsto (-1)^{\ell (\ell+1)/2}  e_{I_0}$, the subgroup PAut$(C^{\bA}(\ell,m)) \subset$ MAut$(C^{\bA}(\ell,m))$ is generated by $\mG^A_1$ and $ (-1)^{\ell (\ell+1)/2}  \tilde \ast_{\ell}$, and hence the group $\tilde H: = (\varrho_1')^{-1}($PAut$(C^{\bA}(\ell,m)))$ is generated by 
the group $H$ (defined above) and the element $\epsilon_1:=( I_{P_{\ell,\ell}/ \mu_{\lambda} },(-1)^{\ell (\ell+1)/2} ) \epsilon$. By definition of $\tilde H$,  $\varrho_1'$ induces an isomorphism $\tilde H/K \simeq$  PAut$(C^{\bA}(\ell,m))$.    The group $\tilde H \subset \tilde G$ has the presentation:
\[ 
\left\langle H, \epsilon_1 \, | \,   \epsilon_1^2 = I_H,  \;\epsilon_1 \left(\bbsm A & {\mathbf{u}}\\ 0 & B \besm \, {\rm mod }\, \mu_{\lambda}, \; {\rm det}(B)^{-1}\right)  \epsilon_1^{-1} =  \left( -\tilde t(\bbsm A & {\mathbf{u}}\\ 0 & B \besm)\, {\rm mod }\, \mu_{\lambda}, {\rm det}(A)\right) \right\rangle .
\]
Now $-\tilde t (\bbsm A & {\mathbf{u}}\\ 0 & B \besm) = \kappa   \bbsm A & {\mathbf{u}}\\ 0 & B \besm^{-t} \kappa^{-1}$ is of the form 
$\bbsm \tilde B & *\\ 0 & \tilde A \besm$, where $\tilde A_{ij} = (A^{-t})_{\ell-i+1, \ell-j+1}$. Since $\tilde A$ is obtained from $A^{-t}$ by
 permuting the rows as well as the columns by the same permutation, we see det$(\tilde A) = $ det$(A)^{-1}$ and hence $\epsilon_1 H \epsilon_1^{-1} = H$. 
The group  $P_{\ell,\ell}/ F^{\times} \rtimes_{-\tilde t} \mathbb Z/ 2 \mathbb Z$ of Definition \ref{Pll2_def} has the presentation:
\[ 
\left\langle P_{\ell,\ell}/ F^{\times}, \epsilon \, |\, \epsilon^2=I_{\PGL(m,F)}, \epsilon f \epsilon^{-1} = \kappa  f^{-t} \kappa^{-1} \right\rangle.
\]
The map from $\tilde H \to P_{\ell,\ell}/ F^{\times} \rtimes_{-\tilde t} \mathbb Z/ 2 \mathbb Z$ that restricts to the epimorphism $H \to P_{m-\ell,\ell}/ F^{\times}$ defined above in the case $m \neq 2 \ell$, and which  sends $\epsilon_1 \mapsto \epsilon$ is a well defined epimorphism with kernel being $K$ again. Therefore $\tilde H/K$ and hence  PAut$(C^{\bA}(\ell,m))$ are isomorphic to  $P_{\ell,\ell}/ F^{\times} \rtimes_{-\tilde t} \mathbb Z/ 2 \mathbb Z$.
\end{proof}

\section{Code associated with the Schubert divisor of $\glm$}
We use the notations of the previous section. Also $\ell, m\in \bZ$ are kept fixed and it will be assumed that $1< \ell < m$. We recall that the Schubert divisor $\omg$ is the intersection  of $\glm$ with the hyperplane $\mH_0$ of $\bP(\wl)$  as defined in Section \ref{sec4}. We throughout assume $m \geq 3$, for otherwise $\omg$ is $\bP^0$ or $\varnothing$. We consider the projective system $\mP_{\omg} = \omg \subset \mH_0$. The fact that the second higher weight $d_2$ of the Grassmann code $C(\ell,m)$ is $q^{\ell(m-\ell) - 1}(1 + q)$ (see \cite{GPP,Nogin}) together with the fact that $|\glv| = q^{\ell(m-\ell)}+|\omg|$, implies that any codimension $2$ subspace of $\bP(\wl)$ intersects $\glv$ in at most $|\omg|  - q^{\ell(m-\ell) - 1}$ points. In particular the projective system $\mP_{\omg}  \subset \mH_0$ is nondegenerate. The Schubert code $C_{\omg}(\ell,m)$ is the linear code associated with this projective system. The study of these codes, and of more general Schubert codes, goes back to \cite{GL}. It may be noted that in the notation of \cite{GL}, $C_{\omg}(\ell,m) = C_{\alpha}(\ell,m)$, where $\alpha\in \mathcal I(\ell,m)$ is given by $\alpha_1 = m-\ell$ and $\alpha_j =  m-\ell+j$ for $2\le j\le \ell$.  In this section we determine the automorphism group of $C_{\omg}(\ell,m)$. 

\subsection{Automorphisms of the Schubert divisor of the Grassmannian}
In this subsection  the Grassmann variety $\Glm$ will be defined over an arbitrary field $\bK$.  We define:
\[  
{\rm Aut}(\Omega) = \{g \in \PgL (\mH_0) :  g(\Omega) = \Omega\}. 
\]
We need a few results before we can determine $\Aut(\omg)$. We begin with an analogue of Lemma \ref{max_lin} for the Schubert divisor $\Omega$. We continue to use the notation introduced in  Lemma \ref{max_lin} and in the preceding and succeeding paragraph. In particular, $\beta$ and $\gamma$ denote elements of $G_{\ell-1}(V)$ and $G_{\ell+1}(V)$, respectively. 
 
\begin{lemma} \label{max_lin_schubert} 
The maximal linear subspaces of $\Omega$ are 
\beqs
&\pi_{\beta}&  \emph{of dimension } \; m-\ell, \quad \emph{where}\;\,\dim(\beta \cap V_{m-\ell}) > 0, \\
&\pi^{\delta}& \emph{ of dimension } \; \ell, \quad \emph{where}\;\,\dim(\delta \cap V_{m-\ell}) > 1, \\
&\tilde \pi_{\beta}:= \pi_{\beta}^{\beta +V_{m-\ell}}& \emph{ of dimension}\; m-\ell-1, \quad \emph{where}\;\, \dim(\beta \cap V_{m-\ell}) =  0, \\
&\tilde \pi^{\delta}= \pi^{\delta}_{\delta \cap V_{m-\ell}}& \emph{ of dimension }\; \ell-1,\quad \emph{where}\;\,\dim(\delta \cap V_{m-\ell}) = 1.
\eeqs
Here $V_{m-\ell}$ is as defined in the beginning of Section \ref{sec4}.
\end{lemma}

\begin{proof} Let $\pi$ be a maximal linear subspace of $\omg$. Since $\pi$ is also a linear subspace of $\glv$, $\pi$ is contained in some $\pi_{\beta} \cap \omg$ or some $\pi^{\delta} \cap \omg$. Suppose $\pi \subset \pi_{\beta} \cap \omg$. If ${\rm dim}(\beta \cap V_{m-\ell}) \geq  1$,  then  $\pi_{\beta} \cap \omg = \pi_{\beta}$ and so $\pi = \pi_{\beta}$,  whereas if ${\rm dim}(\beta \cap V_{m-\ell}) =  0$,  then $\pi_{\beta} \cap \omg = \pi_{\beta}^{\beta +V_{m-\ell}}$ and 
so $\pi=\tilde \pi_{\beta}$. Similarly,  suppose   $\pi \subset \pi^{\delta} \cap \omg$ (where dim$(\delta \cap V_{m-\ell}) \geq 1$ necessarily). If dim$(\delta \cap V_{m-\ell}) \geq  2$, then $\pi^{\delta} \cap \omg = \pi^{\delta}$ and hence $\pi = \pi^{\delta}$,  whereas  if  dim$(\delta \cap V_{m-\ell}) = 1$ then 
$\pi^{\delta} \cap \omg = \pi^{\delta}_{\delta \cap V_{m-\ell}}$ and hence $\pi = \tilde \pi^{\delta}$. 
\end{proof}

As before $(\cap_{\gamma \in \tilde \pi_{\beta}} \gamma) = \beta$ and $(\cup_{\gamma \in \tilde \pi^{\delta}}) = \delta$ show that $\beta \neq \beta'$ implies  $\tilde \pi_{\beta} \neq \tilde \pi_{\beta'}$, and  $\delta \neq \delta'$ implies $\tilde \pi^{\delta} \neq \tilde \pi^{\delta'}$. Moreover in case  $m = 2 \ell$, no $\tilde \pi_{\beta}$ is a $\tilde \pi^{\delta}$, although they are both $(\ell-1)$-dimensional. We recall that  $\pi_{\beta} \cap \pi_{\beta'}$ is the point  $\beta+\beta'$ when  dim$(\beta+\beta') = \ell$, and is empty otherwise. Therefore, $\tilde \pi_{\beta} \cap \tilde \pi_{\beta'}$ is the point  $\beta+\beta'$ when $\beta+\beta' \in W_1 \subset \glv$, and is empty otherwise. Similarly, $\pi_{\delta} \cap \pi_{\delta'}$ is the point $\delta \cap \delta'$  if dim$(\delta \cap \delta') = \ell$, and is empty otherwise. Therefore,  $\tilde \pi_{\delta} \cap \tilde \pi_{\delta'}$ is the point $\delta \cap \delta'$  if $\delta \cap \delta' \in W_1 \in \glv$, and is empty otherwise. Finally, $\pi_{\beta} \cap \pi^{\delta}$ is the line $\pi_{\beta}^{\delta}$ if  $\beta \subset \delta$, and is empty otherwise. Therefore,  $\tilde \pi_{\beta} \cap \tilde \pi^{\delta}$  is the point $\beta \oplus (\delta \cap V_{m-\ell}) \in W_1$ when  $\beta \subset \delta$, and is empty otherwise.  We record this as a lemma for later reference:

\begin{lemma} \label{max_W1_cap}
Let  $ \delta, \delta' \in G_{\ell+1}(V)$ and $\beta, \beta' \in G_{\ell-1}(V)$. Then:
\begin{enumerate}
	\item[{\rm (i)}] 
$\tilde \pi_{\delta} \cap \tilde \pi_{\delta'}$ equals $\delta \cap \delta' \in W_1$  if dim$(\delta \cap \delta') = \ell$, and is empty otherwise. 
\item[{\rm (ii)}]
$\tilde \pi_{\beta} \cap \tilde \pi^{\delta}$  equals $\beta \oplus (\delta \cap V_{m-\ell}) \in W_1$ when  $\beta \subset \delta$, and is empty otherwise.
\item[{\rm (iii)}]
$\tilde \pi_{\beta} \cap \tilde \pi_{\beta'}$ equals  $\beta+\beta'$ when $\beta+\beta' \in W_1 \subset \glv$, and is empty otherwise. 
\end{enumerate}
\end{lemma}

We recall the decomposition $\omg = W_1 \cup \dots \cup W_{\ell}$ defined in \eqref{eq:W_i_def}. The subset $W_1$ is just the Schubert cell associated with the Schubert variety $\omg$ (\cite{GH}). The next few results lead to the fact (see Proposition \ref{W_1_preserve})
that automorphisms of $\omg$ must preserve $W_1$. We define:
\beqs 
W_0^- &=& \{ \beta \in G_{\ell-1}(V)  :  {\rm dim}(\beta \cap V_{m-\ell}) =0 \}, \text{ and } \\
W_1^+ &=& \{ \delta \in G_{\ell+1}(V)  :   {\rm dim}(\delta \cap V_{m-\ell}) =1\}. 
\eeqs

\begin{lemma} \label{lin_W_1}
The maximal linear subspaces of $W_1$ are  
$\tilde \pi_{\beta}$ for $\beta \in W_0^-$ and $\tilde \pi^{\delta}$ for $\delta \in W_1^+$. 
The  $\tilde \pi_{\beta}$ are $(m -\ell-1)$-dimensional, and the  $\tilde \pi^{\delta}$ are $(\ell-1)$-dimensional.
\end{lemma}

\begin{proof}  
The stated subspaces are maximal by the previous lemma. Conversely, let $\pi$ be a maximal  linear subspace of $W_1$,  say of dimension $r$. It suffices to prove that $\pi$ is contained in a $\tilde \pi_{\beta}$ or a $\tilde \pi^{\delta}$. We can assume $r = {\rm dim}(\pi) >0$, for otherwise $\pi$ is a point of $W_1$ and every point is contained in some $\tilde \pi_{\beta}$ by the next lemma. By the previous lemma $\pi$ is contained in a $\pi_{\beta} \cap W_1$ or a $\pi^{\delta} \cap W_1$. 

Suppose $\pi \subset \pi_{\beta} \cap W_1$, and let $\nu_{\beta} := \dim (\beta \cap V_{m-\ell})$. If $\nu_{\beta} = 0$, then $\pi \subset \tilde \pi_{\beta}$, where as if $\nu_{\beta}>1$ then $\pi_{\beta} \cap W_1 = \varnothing$. In the remaining case  $\nu_{\beta}= 1$,  writing $\pi = \pi_{\beta}^U$ for some $U \in G_{\ell+r}(V)$ containing $\beta$,  we must have  $\dim (U \cap V_{m-\ell}) =1$  (for, otherwise $\pi \cap W_2 \neq \varnothing$),  and dim$(U) \leq \ell+1$ (for, otherwise $\dim (U \cap V_{m-\ell}) >1$). Hence dim$(U) = \ell+1$, and  we get $\pi \subset \tilde \pi^{\delta}$ for $\delta = U$. 

Similarly, in the case $\pi \subset \pi^{\delta}$, we consider $\nu_{\delta} := \dim (\delta \cap V_{m-\ell})$. Since $\delta$ is $(\ell+1)$-dimensional, we must have $\nu_{\delta} \geq 1$. If $\nu_{\delta} = 1$, then $\pi \subset \tilde \pi^{\delta}$, where as if $\nu_{\delta}>2$ then $\pi^{\delta} \cap W_1 = \varnothing$. In the remaining case  $\nu_{\delta} = 2$, we write  $\pi = \pi^{\delta}_U$  for some $U \in G_{\ell-r}(\delta)$. Now $\pi \cap W_2 = \varnothing$ forces $\dim (U) \geq \ell-1$ and $\dim (U \cap V_{m-\ell}) =0$. Since dim$(U) = \ell-r$ with $r>0$, we must have dim$(U) = \ell-1$, and we get  $\pi \subset \tilde \pi_{\beta}$ for $\beta = U$.  \end{proof}

\begin{lemma} \label{W_1_union}  The Schubert cell $W_1$ is the union of all the $\tilde \pi_{\beta}$'s, as well as the union of all the $\tilde \pi^{\delta}$'s. In other words, 
\[ 
\bigcup_{\beta \in W_0^-} \tilde \pi_{\beta} = W_1 = 
 \bigcup_{\delta \in W_1^+} \tilde \pi^{\delta}.   
 \]  
\end{lemma}

\begin{proof} 
By definition, the $\tilde \pi_{\beta}$'s and the $\tilde \pi^{\delta}$'s are contained in $W_1$. Conversely, given any  $\gamma \in W_1$, we have  $\dim (\gamma \cap V_{m-\ell}) = 1$. Therefore,  we can find $\beta \subset \gamma$ and $\delta \supset \gamma$ with $\beta \in W_0^-$ and $\delta \in W_1^+$. It then follows that $\gamma \in \tilde \pi_{\beta}$ and $\gamma \in \tilde \pi^{\delta}$.
\end{proof}

\begin{proposition} \label{W_1_preserve}
If $f \in $ \emph{Aut}$(\omg)$ then $f(W_1) = W_1$.
\end{proposition}

\noi
{\bf Remark:} The proposition is motivated by the special case when $\bK$ is algebraically closed. In this case it is known that $W_1$ is the smooth locus of $\omg$ (\cite[Theorem 5.3]{L-W}). Since $\PGL(\mH_0)$ acts smoothly on $\mH_0$, any $f \in $ $\Aut(\omg)$ must preserve $W_1$. 
\begin{proof} We observe that any $f \in$ $\Aut(\omg)$ preserves the set of  maximal linear subspaces of $\omg$ (if $\pi$ is maximal and $\pi' \supsetneq f(\pi)$ then $f^{-1}(\pi') \supsetneq \pi$). We consider three cases: $m  = 2 \ell-1, m = 2 \ell+1$ and $m \neq 2 \ell \pm 1$. If $m=2 \ell-1$, then it follows from Lemma \ref{max_lin_schubert}, that the only maximal linear subspaces of $\omg$ of dimension $\ell-2$ are the $\tilde \pi_{\beta}$'s. (Since $m \geq 3$, we have $\ell-2 \geq 0$, therefore $\tilde \pi_{\beta}$ has at least one point). Therefore the set of $\tilde \pi_{\beta}$'s is preserved, and hence $f(W_1) = W_1$ by Lemma \ref{W_1_union}.
If $m=2 \ell+1$ then, it follows from Lemma \ref{max_lin_schubert}, that the only maximal linear subspaces of $\omg$ of dimension $\ell-1$ are the $\tilde \pi^{\delta}$'s. Therefore the set of $\tilde \pi^{\delta}$'s is preserved, and hence $f(W_1) = W_1$ by Lemma \ref{W_1_union}. Finally  if $m \neq 2 \ell \pm 1$, then the only maximal linear subspaces of dimension $m -\ell -1$ are the $\tilde \pi_{\beta}$'s (and also the $\tilde \pi^{\delta}$'s, if $m = 2 \ell$). Thus $f$ preserves 
the set of $\tilde \pi_{\beta}$'s (or the $\tilde \pi_{\beta}$'s and $\tilde \pi^{\delta}$'s taken together if $m = 2 \ell$).
By Lemma \ref{W_1_union}, $f(W_1) = W_1$.  \end{proof}

For each $\gamma \in W_0$, the decomposition $V = V_{m-\ell} \oplus \gamma$ induces a natural inclusion $\hat \psi:\wedge^{1} V_{m-\ell} \otimes \wedge^{\ell-1} \gamma \hookrightarrow \wlv$ that takes a decomposable element $v \otimes \beta$ to  $v \wedge \beta$. Let $\psi:\bP(V_{m-\ell} \otimes \wedge^{\ell-1} \gamma) \hookrightarrow \bP(\wlv)$ denote the corresponding projective map.  We note that $\psi$ bijectively carries decomposable elements $v \otimes \beta$ to decomposable elements $v \wedge \beta \in W_1 \subset \omg$. Since $\omg \subset \mH_0$, we see im$(\psi)$ is an $\ell(m-\ell)-1$-dimensional linear subspace of  $\mH_0$.
The decomposable elements of $\bP(V_{m-\ell} \otimes \wedge^{\ell-1} \gamma)$ can be identified with $\bP(V_{m-\ell}) \times \bP(\wedge^{\ell-1} \gamma)$ by the Segre embedding. We define:
\beq 
\Delta_{\gamma}:= \psi(\bP(V_{m-\ell}) \times \bP(\wedge^{\ell-1} \gamma)) = {\rm im}(\psi) \cap \omg .
\eeq
Thus the subvariety  $\Delta_{\gamma} \subset \omg $ is an embedding of  $\bP^{m-\ell-1} \times \bP^{\ell-1}$ into $\omg$. 
For each $v \in \bP(V_{m-\ell})$, $\psi(v \otimes \wedge^{\ell-1} \gamma)$ is just $\tilde \pi^{\delta}$  where  $\delta \in \gmv$ is  represented by $v \wedge \gamma$.  For each $\beta \in \bP( \wedge^{\ell-1}\gamma) = G_{\ell-1}(\gamma)$, we have   
$\psi( \bP(V_{m-\ell}) \otimes  \beta)$ is just $\tilde \pi_{\beta}$. Therefore, 
$\Delta_{\gamma}$ can also be described as a union of  maximal linear subspaces of $W_1$:
\beq \label{eq:Delta_gamma} 
\Delta_{\gamma}= \bigcup_{\beta \subset \gamma} \tilde \pi_{\beta} = \bigcup_{\delta \supset \gamma} \tilde \pi^{\delta} \eeq

For $\gamma \in W_0$ and $\beta \neq  \beta' \in G_{\ell-1}(\gamma)$ and $\delta \neq \delta' \in G_{\ell+1}(V)$ containing $\gamma$, we have $\beta+\beta' = \delta \cap \delta'  = \gamma \notin W_1$. Therefore it follows by Lemma \ref{max_W1_cap} that  the unions    
appearing in \eqref{eq:Delta_gamma} are disjoint.

The next lemma and  proposition concern properties of $\Delta_{\gamma}$ which are needed in Theorem \ref{aut_omg}.
\begin{lemma} \label{Delta_lem} 
Let $\gamma, \gamma' \in W_0$.  Then: 
\begin{enumerate} 
\item[{\rm (i)}]  
Each line in $\glv$ through $\gamma$ intersects $\omg$ in a unique point. The set of all such intersection points is precisely $\Delta_{\gamma}$.  \item[{\rm (ii)}]  
$\Delta_{\gamma} = \Delta_{\gamma'}$  if and only if $\gamma= \gamma'$.
\end{enumerate}
\end{lemma}

\begin{proof}
(i)  A line $\pi_{\beta}^{\delta}$ in $\glv$  passes through   $\gamma \in W_0$ if and only if  $\beta \in W_0^-$ and $\delta \in W_0^+$ with $\beta \subset \gamma \subset \delta$. The intersection of  $\pi_{\beta}^{\delta} \cap \omg$  equals $\tilde \pi_{\beta} \cap  \tilde \pi^{\delta} = \beta \wedge v_{\delta}$ where $v_{\delta} \in \bP(V_{m-\ell})$ is the point $\bP(\delta \cap V_{m-\ell})$, as noted in Lemma \ref{max_W1_cap}. Therefore, the set of intersection points of lines through $\gamma$ with $\omg$ is 
\[
\bigcup_{\beta \subset \gamma} \tilde \pi_{\beta}
\] 
and this union is $\Delta_{\gamma}$, thanks to  \eqref{eq:Delta_gamma}.
\smallskip

(ii) Suppose $\Delta_{\gamma} = \Delta_{\gamma'}$. For any $\beta' \in G_{\ell-1}(\gamma')$, we have $\tilde \pi_{\beta'} \subset \Delta_{\gamma'}$, hence 
$\tilde \pi_{\beta'} \subset \Delta_{\gamma}$. Since   $\Delta_{\gamma}$ is the disjoint union of $\tilde \pi^{\gamma \wedge v}$ as $v$ runs over $\bP(V_{m-\ell})$, we must have  $\beta' \wedge v \in  \tilde \pi^{\gamma \wedge v}$, because $v$ is the unique element of $\bP(V_{m-\ell} \cap (\beta \wedge v))$ as well as 
$\bP(V_{m-\ell} \cap (\gamma \wedge v))$. Therefore, $\beta' \subset (\gamma \oplus \bK \hat v)$ for all $\hat v \in V_{m-\ell}\!\setminus\!\{0\}$.
In particular, for two linearly independent vectors $\hat v, \hat w \in V_{m-\ell}$, we see $\beta' \subset  (\gamma \oplus \bK \hat v) \cap (\gamma \oplus \bK \hat w)$ which is $\gamma$ (because $\gamma \cap V_{m-\ell} = \{0\}$). We have thus shown that every $(\ell-1)$-dimensional subspace of $\gamma'$ is also a subspace of $\gamma$, which implies  $\gamma = \gamma'$. Conversely, if $\gamma = \gamma'$, 
then clearly  $\Delta_{\gamma} = \Delta_{\gamma'}$.
\end{proof}

\begin{proposition} \label{Delta_prop}
Let  $X,Y$ and $Z$ be $\bK$-vector spaces of dimensions $m-\ell, \ell$ and  $(m-\ell)\ell$ respectively such that $\bP(Z)$ a linear subspace of  the hyperplane $\mH_0$. 
Let  $\Psi: \bP(X \otimes Y) \to \bP(Z)$  be a projective semilinear isomorphism carrying decomposable elements \emph{(}i.e., $\bP(X) \times \bP(Y)$, by Segre embedding\emph{)} to decomposable elements \emph{(}i.e.,  $\bP(Z) \cap \omg$ \emph{)}. Suppose that on $\bP(X) \times \bP(Y)$, $\Psi$ takes $ [x] \times 
\bP(Y)$  to  a $\tilde \pi^{\delta(x)}$ and $\bP(X) \times [y]$  to a $\tilde \pi_{\beta(y)}$, for each $[x] \in \bP(X)$ and $[y] \in \bP(Y)$. Then $\Psi(\bP(X) \times \bP(Y)) = \Delta_{\gamma}$ for some $\gamma \in W_0$.
\end{proposition}
\begin{proof} 
Let $[x], [y], \delta(x)$ and $\beta(y)$ be as in the hypothesis. Let $\Psi_1(x)$ denote the point $\bP( \delta(x) \cap V_{m-\ell})$ of $\bP(V_{m-\ell})$. Now, any element $\gamma' \subset \tilde \pi_{\beta(y)}$ can be written uniquely as $v \wedge \beta(y)$ where $v$ is the point $\bP( \gamma' \cap V_{m-\ell})$ of $\bP(V_{m-\ell})$.  Let $\gamma':=\Psi([x \otimes y]) \in \tilde \pi^{\delta(x)}$. Then $\Psi_1(x)$ is $\bP(\gamma' \cap V_{m-\ell})$, hence $\Psi([x \otimes y]) = \Psi_1(x) \wedge \beta(y)$. 
Since $\Psi_{|(\bP(X) \times [y])}:\bP(X) \to \tilde \pi_{\beta(y)}$ is a bijection, we  conclude that $\Psi_1:\bP(X) \to \bP(V_{m-\ell})$ is a bijection. Let $[y] \neq [y']$, then $\bP(X) \times [y]$ and $\bP(X) \times [y']$ are disjoint, and since $\Psi$ is injective we get $\tilde \pi_{\beta(y)}$ and $\tilde \pi_{\beta(y')}$ are disjoint.  In particular, we get $\beta(y) \neq \beta(y')$ and hence $\ell \leq \dim (\beta(y) + \beta(y'))$. Since $\Psi_1(x) \wedge \beta(y), \Psi_1(x) \wedge \beta(y') \in  \tilde \pi^{\delta(x)}$  for each $[x]  \in \bP(X)$, both $\beta(y)$ and $\beta(y')$ are contained in $\delta(x)$, and hence $\dim (\beta(y) + \beta(y')) \leq \ell+1$. In case this dimension is $\ell+1$, it follows that $\beta(y) + \beta(y') =  \delta(x)$ for all $[x]  \in \bP(X)$. In particular $\Psi_1(x) = \bP(\delta(x) \cap V_{m-\ell})$ is independent of $[x]$ contradicting the fact proved above that $\Psi_1$ is bijective. Therefore, $\gamma:=\beta(y) + \beta(y')$ is $\ell$-dimensional, and the fact that dim$(\beta(y) \cap V_{m-\ell}) = 0$   implies dim$(\gamma \cap V_{m-\ell}) \leq 1$, i.e $\gamma \in W_0$ or $W_1$. Since $\tilde \pi_{\beta(y)}$ and $\tilde \pi_{\beta(y')}$ are disjoint, it follows by  Lemma  \ref{max_W1_cap} that $\gamma \in W_0$.  Writing $\delta(x) = \gamma \wedge \Psi_1(x)$ for each $[x] \in \bP(X)$, we see that $\{\delta(x)  :  x \in \bP(X)\}$ is the family of all $\delta \in \gmv$ containing $\gamma \in W_0$. Therefore:
\[ \Psi(\bP(X) \times \bP(Y)) = \bigcup_{x \in \bP(X)} \Psi([x] \times \bP(Y)) = \bigcup_{x \in \bP(X)} \tilde \pi^{\delta(x)}  = \bigcup_{\delta \supset \gamma} \tilde \pi^{\delta}\] which by \eqref{eq:Delta_gamma} equals $\Delta_{\gamma}$. 
\end{proof}

The next lemma is needed for the case $m = 2 \ell $ in Theorem \ref{aut_omg}.

\begin{lemma} \label{atmost_one_line} 
Let $\gamma \in W_0$, $\beta \in W_0^-$, and $\delta \neq  \delta'  \in W_0^+$ 
\begin{enumerate}
\item[{\rm (i)}]   
If $\tilde \pi_{\beta}$ and $\tilde \pi^{\delta}$ intersect, then the only lines in $\omg$  joining a point of $\tilde \pi_{\beta}$ to a point of $\tilde \pi^{\delta}$  are the ones passing through the point $\tilde \pi_{\beta} \cap \tilde \pi^{\delta}$ and completely contained in $\tilde \pi_{\beta}$ or $\tilde \pi^{\delta}$.
\item[{\rm (ii)}]  
If $\tilde \pi^{\delta}$ and $\tilde \pi^{\delta'}$ intersect, then there is a line in $\omg$ joining a point of $\tilde \pi^{\delta}$ to a point of $\tilde \pi^{\delta'}$ which is not completely contained in $W_1$. 
\end{enumerate}
\end{lemma} 

\begin{proof} 
{(i)} Let  $\bK v_0 =  V_{m-\ell} \cap \delta$. As observed in the discussion after Lemma \ref{max_lin_schubert}, 
$\tilde \pi_{\beta} \cap \tilde \pi^{\delta} = \bK v_0 \oplus \beta$. Now, let $L$ be a line joining a point $\gamma_1 \in \tilde \pi_{\beta}$ to some point $\gamma_2 \in \tilde \pi^{\delta}$. Let $\bK v_1 = V_{m-\ell} \cap \gamma_1$. If $v_0$ and $v_1$ are dependent, then clearly $\gamma_1  =  \bK v_0 \oplus \beta=\tilde \pi_{\beta} \cap \tilde \pi^{\delta}$, and therefore $L \subset   \tilde \pi^{\delta}$. If $v_0$ and $v_1$ are independent, then we pick a complement $\gamma_0$ of $\bK v_0$ in $\delta$ satisfying $\beta \subset \gamma_0$. Every element of $\tilde \pi^{\delta}$ (and hence $\gamma_2$ in particular) is of the form $\bK v_0 \oplus \beta'$ for some $(\ell-1)$-dimensional subspace $\beta'$ of $\gamma_0$. We write $\gamma_1+\gamma_2 = (\bK v_0 \oplus \bK v_1) \oplus (\beta + \beta')$ (where we have used the fact that $(\beta+\beta') \cap V_{m-\ell} \subset  \gamma_0 \cap  V_{m-\ell} = \{0\}$).  By definition of a line in $\glv$, $\dim (\gamma_1 + \gamma_2)=\ell+1$, and hence $\dim (\beta+\beta') = \ell-1$, which implies $\beta = \beta'$. Thus  $\gamma_2 \in \tilde \pi_{\beta}$ and hence $L \subset \tilde \pi_{\beta}$. 
\smallskip

{(ii)}  By definition of $\tilde \pi^{\delta}$, it follows that $\tilde \pi^{\delta}$ and $\tilde \pi^{\delta'}$ intersect if and only if $ \delta \cap V_{m-\ell} = \delta' \cap V_{m-\ell} = \bK v_0$ (for some $v_0 \in V_{m-\ell})$, and $\gamma:=\delta \cap \delta' \in W_1 \subset \glv$. We write $\delta = \gamma \oplus \bK u$ and $\delta' = \gamma \oplus \bK u'$. Since $\delta +\delta'$ is $(\ell+2)$-dimensional there is a $v_1 \in V_{m-\ell}$ contained in $\delta+\delta'$ and independent of $v_0$. We write $v_1 = u'' + a u +b u'$ for some $u'' \in \gamma$ and $a,b \in \bK$. Since $v_1 \notin \delta, \delta'$, the scalars $a, b$ are both non-zero. Hence we can express $\delta = \gamma \oplus \bK u'''$ where $u''' = u''+a u$. Now for any $(\ell-1)$-dimensional subspace $\beta$ of $\gamma$ containing $v_0$, we observe that the line joining $\beta \oplus \bK u''' \in \tilde \pi^{\delta}$ and $\beta \oplus \bK u' \in \tilde \pi^{\delta'}$  is not contained in $W_1$ as it contains $\beta \oplus \bK v_1 \in W_2$. 
\end{proof}

\begin{theorem} \label{aut_omg}
The automorphism group of the Schubert divisor $\Omega$ is given by 
\[
\Aut (\Omega) \simeq \Aut(W_0) =  \begin{cases} \rho( \hat P_{m-\ell,\ell}/\bK^{\times} ) & \mbox{ if } \; m \neq 2 \ell, \\
\rho( \hat P_{\ell,\ell}/ \bK^{\times})\rtimes_{\tilde \ast_{\ell}} \mathbb Z/ 2 \mathbb Z  & \mbox{ if } \; m = 2 \ell. \end{cases}
\]
More precisely,  for each $f \in$ \emph{Aut}$(\omg)$, there is a unique way to extend $f$ to an automorphism $\iota(f)$ of $\glv$. The image of this monomorphism from $\Aut (\Omega)$ to $\Aut(\glv)$ is the subgroup $\Aut(W_0)$ of $\Aut(\glv)$.
\end{theorem}

\begin{proof} 
By Theorem \ref{aut_w0}, $\Aut(W_0)$ is the subgroup of $\Aut(\glv)$ that preserves $W_0$, and hence its complement $\omg$. Every element of $\Aut(W_0)$  preserves $\mH_0$: Writing $f \in \Aut(W_0)$  as $\wedge^{\ell} g$  (or also $(\wedge^{\ell} g) \tilde \ast_{\ell}$ if  $m = 2 \ell$) where $\hat g(x) = \bbsm A & {\mathbf{u}}\\ 0 & B \besm \mu(x)$ is a representative of $g$ in $\gL(V)$, we see that $p_{I_0}( \wedge^{\ell} \hat g (\xi)) = \mu (p_{I_0} (\xi)) {\rm det}(B)$. Therefore $\wedge^{\ell} g$ preserves $\mH_0$.  In case $m = 2 \ell$, it was shown in the discussion following Definition \ref{Pll2_def} that $\tilde \ast_{\ell} \in$ Aut$(W_0)$ preserves $\mH_0$. Therefore, we have a restriction homomorphism:
\[ 
{\rm res}_{\omg}:\Aut(W_0) \to \Aut(\omg) \quad \text{given by } \quad f \mapsto f_{|\mH_0}. 
\]
We claim that this homomorphism is injective, i.e.,  if $f \in \Aut(W_0)$ fixes  $\omg$ pointwise, then it fixes $W_0$ pointwise. We will show $f(\gamma) = \gamma$ for each  $\gamma \in W_0$.  By Lemma \ref{Delta_lem}, the set of intersection points of all lines through $\gamma \in W_0$ with $\omg$ is $\Delta_{\gamma}$. Since $f$ is also an automorphism of $\glv$ (by Theorem \ref{aut_w0}), these lines get mapped to the set of all lines through $f(\gamma) \in W_0$. Since $f$ fixes $\omg$ pointwise, we see that the intersection points with $\omg$ of all lines through $f(\gamma)$ is $\Delta_{\gamma}$, however it is also $\Delta_{f(\gamma)}$ by Lemma \ref{Delta_lem}. Thus, we get $\Delta_{\gamma} = \Delta_{f(\gamma)}$, and part (ii) of Lemma \ref{Delta_lem}  implies $\gamma = f(\gamma)$. 

In order to show  ${\rm res}_{\omg}$ is surjective, we will construct an extension of each  $f \in\Aut(\omg)$ to  an  automorphism $\iota(f)$ of $\glv$ that preserves $W_0$, i.e.,  such that ${\rm res}_{\omg} (\iota(f)) = f$ for all $f \in \Aut(\omg)$. Consequently, ${\rm res}_{\omg}$ is an isomorphism and the function $f \mapsto \iota(f)$ will be the inverse isomorphism to  ${\rm res}_{\omg}$. By Proposition \ref{W_1_preserve}, each $f \in \Aut(\omg)$ satisfies $f(W_1) = W_1$, and hence $f$ preserves the set of  maximal linear subspaces of $W_1$. First we assume that no $f(\tilde \pi_{\beta})$ is a $\tilde \pi^{\delta}$ (this is automatic if $m \neq 2 \ell$, by looking at their dimensions). Thus we have bijections $f_-:W_0^- \to W_0^-$ and $f_+:W_0^+ \to W_0^+$ defined by $f(\tilde \pi_{\beta}) = \tilde \pi_{f_-(\beta)}$ and $f(\tilde \pi^{\delta}) = \tilde \pi^{f_+(\delta)}$.  We now fix a $\gamma  \in W_0$, and recall the linear isomorphism  $\psi: \bP( V_{m-\ell} \otimes \wedge^{\ell-1} \gamma) \to {\rm im}(\psi) \subset \mH_0$ (defined after Proposition \ref{W_1_preserve}). Let $X = V_{m-\ell}, Y =\wedge^{\ell-1} \gamma$, and let $Z$ be defined by $\bP(Z) = {\rm im}(f  \circ \psi)$. Since im$(\psi) \subset \mH_0$ and $f$  is a semilinear isomorphism of $\mH_0$, the composition $\Psi:=f \circ  \psi: \bP(X \otimes Y) \hookrightarrow \bP(Z) \subset \mH_0$, is a projective semilinear isomorphism, satisfying $\Psi(v \otimes \wedge^{\ell-1} \gamma) = \tilde \pi^{f_+(\delta)}$ (where $\delta=v \wedge \gamma$), and $\Psi( V_{m-\ell} \otimes \hat \beta) = \tilde \pi_{f_-(\beta)}$ for each $v \in \bP(V_{m-\ell})$ and $\beta \in \bP( \wedge^{\ell-1}\gamma)$.  Therefore, $\Psi$ satisfies all the hypothesis of Proposition  \ref{Delta_prop}, and hence $\Psi(\bP(X) \times \bP(Y))  = f(\Delta_{\gamma})$ is  a $\Delta_{f_0(\gamma)}$ for some $f_0(\gamma) \in W_0$. This defines a function $f_0:W_0 \to W_0$ which has an inverse namely the corresponding function $(f^{-1})_0$. Thus $f_0$ is a bijection. We now have a bijection $\iota(f): \glv \to \glv$ defined by the pair $f, f_0$. In order to show $\iota(f) \in \Aut(\glv)$ we must show that $\iota(f)$ and $\iota(f^{-1})$ carry lines to lines.

Let $L$ be a line in $\glv$. If $L$ does not intersect $W_0$, then $L \subset \omg$ and hence $f(L)$ and $f^{-1}(L)$ are lines in $\omg$. If $L$ passes through a point $\gamma \in W_0$, then writing $L = \pi_{\beta}^{\delta}$ we must have $\beta \subset \gamma \subset \delta$ with dim$(\delta \cap V_{m-\ell})=1$. Now, except the point $\tilde \pi_{\beta}^{\delta}: = \tilde \pi_{\beta} \cap \tilde \pi^{\delta}$, the remaining points of $\pi_{\beta}^{\delta}$ are all in $W_0$.  For each $\gamma' \in \pi_{\beta}^{\delta} \cap W_0$, it follows from the decomposition \eqref{eq:Delta_gamma} of $\Delta_{\gamma'}$ that $\tilde \pi_{\beta}, \tilde \pi^{\delta} \subset \Delta_{\gamma'}$. Hence $\tilde \pi_{f_-(\beta)}, \tilde \pi^{f_+(\delta)} \subset \Delta_{f_0(\gamma')}$, thus proving that $f_-(\beta)  \subset f_0(\gamma') \subset f_+(\delta)$, i.e.,  $f_0(\pi_{\beta}^{\delta} \cap W_0) \subset  \pi_{f_-(\beta)}^{f_+(\delta)}$. Moreover $f$ carries the remaining point of  $\pi_{\beta}^{\delta}$, namely $\tilde \pi_{\beta} \cap \tilde \pi^{\delta}$ to $\tilde \pi_{f_-(\beta)} \cap \tilde \pi^{f_+(\delta)} \in \pi_{f_-(\beta)}^{f_+(\delta)}$. Thus $\iota(f)$ takes the line $\pi_{\beta}^{\delta}$ to the line $\pi_{f_-(\beta)}^{f_+(\delta)}$. Repeating this argument for $f^{-1}$, we conclude that $\iota(f), \iota(f^{-1})$ are line preserving bijections of $\glv$ (which preserve $W_0$), hence $\iota(f) \in \Aut(W_0)$.  

It remains to consider the case when  $m = 2 \ell$ and there exists a $\beta_0$ with $f(\tilde \pi_{\beta_0}) = \tilde \pi^{\delta_0}$. In this case,  for each $\gamma \in W_0$ with $\beta_0 \subset \gamma$ we will prove that $\{f (\tilde \pi_{\beta}) : \beta \subset \gamma\}$ are all $\tilde \pi^{\delta}$'s, and $\{f(\tilde \pi^{\delta}) : \delta \supset \gamma\}$ are all  $\tilde \pi_{\beta}$'s.   Using a connectivity argument (as in the proof of Chow's theorem) we will then show that $f$ interchanges the sets $\{\tilde \pi_{\beta} : \beta \in W_0^-\}$ and  $\{\tilde \pi^{\delta} : \delta \in W_1^+\}$. To begin with, suppose $f(\tilde \pi_{\beta_0}) = \tilde \pi^{\delta_0}$.  Let $v_0 = \bP(\delta_0 \cap V_{m-\ell})$, and pick an arbitrary  $\gamma \in W_0$ containing $\beta_0$. For any $\beta' \subset \gamma$ with $\beta' \neq \beta_0$, suppose $f(\tilde \pi_{\beta'}) = \tilde \pi_{\beta''}$. Since $f(\tilde \pi_{\beta'}) = \tilde \pi_{\beta''}$, there is a $v_0'  \in \bP(V_{m-\ell})$, such that $f(v_0' \wedge \beta') = v_0 \wedge \beta''$. We consider $\tilde \pi^{\delta'}$ where $\delta' = v_0' \wedge \gamma$. Let $L$ be the line in $\tilde \pi^{\delta'}$ joining $v_0' \wedge \beta'$ and $v_0' \wedge \beta_0$. The line $f(L)$ joins $v_0 \wedge \beta''$ to some point $v_0 \wedge \beta''' \in \tilde \pi^{\delta_0}$. In particular every $\gamma''  \in f(L)$ contains $v_0$ (because each $\gamma'' \supset (v_0 \wedge \beta'') \cap(v_0 \wedge \beta''')$). Since the intersection of all $\gamma''$  lying on a line in a $\tilde \pi_{\beta}$ is $\beta$, and $\bP(\beta \cap V_{m-\ell}) =\varnothing$, we see that   $f(\tilde \pi^{\delta'})$ must be  a $\tilde \pi^{\delta''}$ for some $\delta'' \in  $. Since $v_0 \wedge \beta'''$ is a common point of $\tilde \pi^{\delta''}$ and $\tilde \pi^{\delta_0}$, part (ii) of Lemma  \ref{atmost_one_line} implies that there is a line joining a point of $\tilde \pi^{\delta''}$ to a point of $\tilde \pi^{\delta_0}$ which is not contained in $W_1$. However every such line is the image under $f$ of a line
joining a point of $\tilde \pi^{\delta'}$ to a point of $\tilde \pi_{\beta_0}$, which by part (i) of Lemma  \ref{atmost_one_line} is contained in $W_1$. This contradicts the fact that $f$ preserves $W_1$, and therefore we conclude that $f$ carries all $\tilde \pi_{\beta}$'s in $\Delta_{\gamma}$ to $\tilde \pi^{\delta}$'s. Now, suppose there is a $\tilde \pi^{\delta}$ in $\Delta_{\gamma}$ such that $f(\tilde \pi^{\delta})$ is a $\tilde \pi^{\delta'}$. Now  $\tilde \pi_{\beta_0}$ and $\tilde \pi^{\delta}$  intersect in  a point, and hence their images under $f$, namely $\tilde \pi^{\delta_0}$ and $\tilde \pi^{\delta'}$ also intersect in a point. Again using parts (ii) and (i) of Lemma \ref{atmost_one_line}, there is a line joining a point  of $\tilde \pi^{\delta_0}$ to a point of $\tilde \pi^{\delta'}$ which is not contained in $W_1$, but it is the image under $f$ of a line contained in $W_1$, which  contradicts the fact that $f$ preserves $W_1$, and therefore we conclude that $f$ carries all $\tilde \pi^{\delta}$'s in $\Delta_{\gamma}$ to $\tilde \pi_{\beta}$'s. In the reasoning above,  $\gamma \in W_0$ containing $\beta_0$ was arbitrary. Therefore, in order to show that
$f$ interchanges  the sets $\{\tilde \pi_{\beta} : \beta \in W_0^-\}$ and  $\{\tilde \pi^{\delta} : \delta \in W_1^+\}$, it suffices to show that given $\gamma, \gamma'' \in W_0$ there is a sequence $\gamma =\gamma_0, \gamma_1, \dots, \gamma_{\ell} = \gamma''$ such that  $\gamma_{i-1}$ and $\gamma_i$ intersect in a $\beta_i \in W_0^-$.  Now, every $\gamma \in W_0$ has a unique expression of the form $\gamma = (e_{m-\ell+1} +v_1)\wedge \dots \wedge (e_m + v_{\ell})$ where $v_1, \dots, v_{\ell} \in V_{m-\ell}$. Given $\Delta_{\gamma''}$ with $\gamma'' = (e_{m-\ell+1} +v_1')\wedge \dots \wedge (e_m + v_{\ell}')$, let $\gamma_0 = \gamma$, $\gamma_{\ell} = \gamma''$, then the desired sequence  is $\gamma_i := (e_{m-\ell+1} +v_1')\wedge \dots \wedge (e_{m-\ell+i} + v_{i}') \wedge (e_{m-\ell+i+1} + v_{i+1}) \wedge \dots \wedge (e_m + v_{\ell})$.

Now $\tilde \ast_{\ell} \in \Aut (\glv)$ interchanges the set of $\pi_{\beta}$'s and the set of $\pi^{\delta}$'s. Since  $\tilde \ast_{\ell}$ preserves $\omg$ and 
$\tilde \pi_{\beta} = \pi_{\beta} \cap \omg$,  $\tilde \pi^{\delta} = \pi^{\delta} \cap \omg$ (as observed in the proof of Lemma \ref{max_lin_schubert})
we conclude that  $\tilde \ast_{\ell} \in \Aut(\omg)$ interchanges the set of $\tilde \pi_{\beta}$'s and the set of $\tilde \pi^{\delta}$'s. If $f \in \Aut(\omg)$ 
 interchanges the set of $\tilde \pi_{\beta}$'s and the set of $\tilde \pi^{\delta}$'s, then $\tilde \ast_{\ell} \circ f$ preserves the  set of $\tilde \pi_{\beta}$'s as well  as the set of $\tilde \pi^{\delta}$'s, and therefore  $\tilde \ast_{\ell}^{-1} \circ \iota(\tilde \ast_{\ell} \circ f)$ is the desired extension $\iota(f)$ of $f$ to  Aut$(W_0)$. This concludes the construction of the isomorphism $f \mapsto \iota(f)$ between $\Aut(\omg)$ and $\Aut(W_0)$.
\end{proof} 

\noindent
{\bf Remark:} Mark Pankov has brought to our attention the work of Pra\.{z}mowski and \.{Z}ynel \cite{Praz} which shows that automorphisms of certain subspaces of $\glv$ called  \emph{spine spaces} (see \cite[Section 3.6]{Pankov} for a definition) extend to automorphisms of $\glv$. The spaces $W_i$ defined in \eqref{eq:W_i_def} are examples of spine spaces. 

\begin{corollary}  For the Schubert divisor code $C_{\omg}(\ell,m)$, we have 
\[
\Aut(C_{\omg}(\ell,m))/F^{\times}  \simeq  \begin{cases} \rho( \hat P_{m-\ell,\ell}/F^{\times} ) & \mbox{ if } \; m \neq 2 \ell,  \\
\rho( \hat P_{\ell,\ell}/ F^{\times})\rtimes_{\tilde \ast_{\ell}} \mathbb Z/ 2 \mathbb Z  & \mbox{ if } \; m = 2 \ell. \end{cases}
\] 
The groups $\Aut(C_{\omg}(\ell,m))$ and   $\MAut(C_{\omg}(\ell,m))$ are isomorphic to $\Aut(C^{\bA}(\ell,m))$  and $\MAut(C^{\bA}(\ell,m))$, respectively, and hence they can be explicitly described as in  Theorem \ref{thm:AutAff}. 
\end{corollary}

\begin{proof}
Since $\Aut(C_{\omg}(\ell,m))/F^{\times}$ is just $\Aut(\omg)$, the first assertion of the corollary follows from Theorem \ref{aut_omg}.  Let $\hat \mH_0$ be the hyperplane in $\wlv$ given by the vanishing of the  Pl\"{u}cker coordinate $p_{I_0}$ where $I_0:=(m-\ell+1, m-\ell+2, \dots,m)$. Let $p_2: \gL(\hat \mH_0) \to \PgL (\mH_0)$ and $p_1: \gL(\wlv) \to \PgL (\wlv)$ be the projection homomorphisms. The group $\Aut(C_{\omg}(\ell,m))$ is isomorphic to $p_2^{-1}(\Aut(\omg))$, whereas the group $\Aut(C^{\bA}(\ell,m))$ is isomorphic to $p_1^{-1}(\Aut(W_0))$. Given any $f \in \Aut(\omg)$, let $\hat f \in p_1^{-1}(\iota(f))$. Since $\iota(f)$ preserves $\mH_0$, we get a homomorphism  $\jmath: \Aut(C^{\bA}(\ell,m))  \to \Aut (C_{\omg}(\ell,m))$ obtained by restricting each $\hat f$ to $\hat \mH_0$. Now any  $g \in p_2^{-1}(f)$ is of the form $c \, \jmath(\hat f)$ for some non-zero scalar $c$. Since $c \, \jmath(\hat f) =  \jmath(c \, \hat f)$, it follows that $\jmath$ is surjective. Suppose $\hat f \in $ ker$(\jmath)$, then $f = p_2 \circ \jmath(\hat f)$ is the identity of $\Aut(\omg)$, and hence $\iota(f)$ is the identity of $\Aut(W_0)$, whence $\hat f$ is $c$ times the identity of $\gL(\wlv)$. Therefore, $\jmath( \hat f)$ is $c$ times the  identity of $\gL(\hat \mH_0)$ which implies $c = 1$, i.e., $\jmath$ is injective. Thus the homomorphism $\jmath$ is an isomorphism between $\Aut(C_{\omg}(\ell,m))$ and  $\Aut(C^{\bA}(\ell,m))$. Moreover, since the isomorphism $\jmath$   
simply maps elements of $\Aut(C^{\bA}(\ell,m))$ to their restrictions to the hyperplane $\hat \mH_0$, it follows that $\MAut(C_{\omg}(\ell,m))$ is isomorphic to $\MAut(C^{\bA}(\ell,m))$ as well.  \end{proof}

We remark that $\Aut(C(\ell,m))$  is transitive in the sense that given any pair $e_i \neq e_j$ of basic vectors of $F^n$, there is an element of $\Aut(C(\ell,m))$which sends $e_i$ to a scalar multiple of $e_j$. In contrast, $\Aut(C_{\omg}(\ell,m))$ is not transitive since it preserves the basic vectors representing $W_i$, for each $i = 1, \dots ,\ell$. As remarked before, in the case when the field $F$ is replaced by an algebraically closed field, this lack of transitivity is reflected in the fact that $\omg$ is not a smooth variety, and the automorphisms must preserve the smooth locus $W_1 \subset \omg$.

\section*{Acknowledgement}
We are grateful to Mark Pankov for his comments on an initial version of this article and for bringing \cite{Pankov}, \cite{Wan} and \cite{Praz} to our attention. 
Krishna Kaipa would like to thank the Department of Mathematics at IIT Bombay where this work was carried~out.
\bibliographystyle{plain}
\bibliography{refs_SRG}
\nocite{}
\end{document}